\documentclass[11pt]{article}

\usepackage{hyperref}
\usepackage{amsfonts}
\usepackage{graphicx,subfig}
\usepackage{amsmath,amsthm}
\usepackage{wrapfig}
\usepackage{amssymb}

\usepackage{epsfig}
\usepackage{amsmath}
\usepackage{amssymb}
\usepackage{psfrag}
\usepackage[absolute]{textpos}
\usepackage{setspace}

\newtheorem{thm}{Theorem}[section]
\newtheorem{lem}{Lemma} [section]
\newtheorem{definition}{Definition}

\newtheorem{alg}{Algorithm}

\newcommand{\A}{{\bf A}}
\newcommand{\M}{{\bf M}}
\def\x{{\bf x}}
\def\y{{\bf y}}
\def\w{{\bf w}}

\def\t{{\bf t}}
\def\n{{\bf n}}
\def\k{{\bf k}}

\def\z{{\bf z}}
\def\e{{\bf e}}

\def\F{{\mathcal{F}}}
\def\G{{\mathcal{G}}}
\def\W{{\bf W}}
\newcommand{\beq}{\begin{equation}}
\newcommand{\eeq}{\end{equation}}
\newcommand{\bea}{\begin{eqnarray}}
\newcommand{\eea}{\end{eqnarray}}

\newcommand{\Prob}{\ensuremath{\mathbb{P}}}

\long\def\symbolfootnote[#1]#2{\begingroup%
\def\thefootnote{\fnsymbol{footnote}}\footnote[#1]{#2}\endgroup}

\begin{document}

\title{Analyzing Weighted $\ell_1$ Minimization for Sparse Recovery with Nonuniform Sparse Models\footnote{The results of this paper were presented in part at the International Symposium on Information Theory, ISIT 2009.}}
\author{ M. Amin Khajehnejad$^\dag$, Weiyu Xu$^\dag$, A. Salman Avestimehr$^*$ and Babak Hassibi$^\dag$ \\\text{}\\$^\dag$California Institute of Technology, Pasadena CA 91125 \\ $^*$Cornell University, Ithaca NY 14853}
 \maketitle
 \vspace{1cm}
\begin{abstract}
In this paper we introduce a nonuniform sparsity model and analyze
the performance of an optimized weighted $\ell_1$ minimization over
that sparsity model.   In particular, we focus on a model where the
entries of the unknown vector fall into two sets, with entries of
each set having a specific probability of being nonzero. We propose
a weighted $\ell_1$ minimization recovery algorithm and analyze its
performance using a Grassmann angle approach. We compute explicitly
the relationship between the system parameters-the weights, the
number of measurements, the size of the two sets, the probabilities
of being nonzero- so that when i.i.d. random Gaussian measurement
matrices are used, the weighted $\ell_1$ minimization recovers a
randomly selected signal drawn from the considered sparsity model
with overwhelming probability as the problem dimension increases.
This allows us to compute the optimal weights. We demonstrate
through rigorous analysis and simulations that for the case when the
support of the signal can be divided into two different subclasses
with unequal sparsity fractions, the optimal weighted
$\ell_1$ minimization outperforms the regular $\ell_1$ minimization
substantially. We also generalize the results to an arbitrary number of classes.
\end{abstract}
\thispagestyle{empty} 


\setcounter{page}{1}

\section{Introduction}
\label{sec:Intro}

Compressed sensing is an emerging technique of joint sampling and
compression that has been recently proposed as an alternative to
Nyquist sampling (followed by compression) for scenarios where
measurements can be costly~\cite{rice}. The whole premise is that
sparse signals (signals with many zero or negligible elements over a
known basis) can be recovered with far fewer measurements than the
ambient dimension of the signal itself. In fact, the major
breakthrough in this area has been the demonstration that
$\ell_1$ minimization can efficiently recover a sufficiently sparse
vector from a system of underdetermined linear equations \cite{CT}.
$\ell_1$ minimization is usually posed as the convex relaxation of
$\ell_0$ minimization which solves for the sparsest solution of a
system of linear equation and is NP hard.

The conventional approach to compressed sensing assumes no prior
information on the unknown signal other than the fact that it is
sufficiently sparse over a particular basis. In many applications,
however, additional prior information is available. In fact, in many
cases the signal recovery problem that compressed sensing  addresses
is a detection or estimation problem in some statistical setting.
Some recent work along these lines can be found in \cite{Baraniuk
Detection}, which considers compressed detection and estimation,
\cite{Bayesian CS}, which studies Bayesian compressed sensing, and
\cite{Chevher} which introduces model-based compressed sensing
allowing for model-based recovery algorithms. In a more general
setting, compressed sensing may be the inner loop of a larger
estimation problem that feeds prior information on the sparse signal
(e.g., its sparsity pattern) to the compressed sensing algorithm
\cite{Candes Reweighted,ISIT10 Reweighted}.

In this paper we will consider a particular model for the sparse
signal where the entries of the unknown vector fall into a number $u$ of classes, with each class having a specific fraction of nonzero entries. The standard compressed sensing model
is therefore a special case where there is only one class. As mentioned above, there are many situations where such
prior information may be available, such as in natural images,
medical imaging, or in DNA microarrays. In the DNA microarrays
applications for instance, signals are often {\em block
  sparse}, i.e., the signal is more likely to be nonzero in certain
blocks rather than in others \cite{Mihailo BS-CS}. While it is
possible (albeit cumbersome) to study this model in full generality,
in this paper we will focus on the case where the entries of the
unknown signal fall into a fixed number $u$ of categories; in the
$i$th set $K_i$ with cardinality $n_i$, the fraction of nonzero entries is $p_i$. This model is rich enough to capture
many of the salient features regarding prior information. We refer
to the signals generated based on this model as \emph{nonuniform
sparse} signals.

A signal  generated based on this model could resemble the vector
representation of a natural image in the domain of some linear
transform (e.g.  Discrete Fourier Transform, Discrete Cosine
Transform, Discrete Wavelet Transform, ...) or the spatial
representation of some biomedical image, e.g., a brain fMRI image.
Although a brain fMRI image is not necessarily sparse, the
subtraction of the brain image at any moment during an experiment
from an initial background image of inactive brain mode is indeed a
sparse signal which, demonstrates the additional brain activity
during the specific course of experiment. Moreover, depending on the
assigned task, the experimenter might have some prior information.
For example it might be known that some regions of the brain are
more likely to be entangled  with the decision making process than
the others. This can be captured in the above \emph{nonuniform
sparse } model by considering a higher value $p_i$ for the more
active region. Similarly, this model is applicable to other problems
like network monitoring (see ~\cite{NetMon} for an application of
compressed sensing and nonlinear estimation in compressed network
monitoring), DNA microarrays ~\cite{DNAref,DNAref1,DNAref2},
astronomy, satellite imaging and many more practical examples.

In this paper we first analyze this model for the case where there are
$u\geq 2$ categories of entries, and demonstrate through
rigorous analysis and simulations that  the recovery performance can
be significantly boosted by exploiting the additional information.
We find a closed form expression for the recovery threshold for
$u=2$. We also generalize the results to the case of $u >2$. A further interesting question to be addressed in future work would be to characterize the gain in recovery percentage as a function of the number of
distinguishable classes $u$. It is worth mentioning that a somewhat
similar model for prior information has been considered in
\cite{Vas}. There, it has been assumed that part of the support is
 completely known a priori or due to previous processing. A modification of the regular
 $\ell_1$ minimization based on the given information is proven to
 achieve significantly better recovery guarantees. As will be discussed, this
 model can be cast as a special case of the nonuniform sparse model, where the sparsity fraction is equal to unity in one of the classes . Therefore, using the generalized tools of this work, we can explicitly find the recovery thresholds for the method proposed in \cite{Vas}. This is in contrast to the recovery guarantees of \cite{Vas} which are given in terms of the restricted isometry property (RIP).

The contributions of the paper are the following. We propose a
weighted $\ell_1$ minimization approach for sparse recovery where
the $\ell_1$ norms of different classes ($K_i$'s) are assigned different
weights $w_{K_i}$ ($1\leq i\leq u$). Clearly, one would want to give a
larger weight to the entries with a higher chance of being zero and
thus further force them to be zero.\footnote{A somewhat related
method that uses
  weighted $\ell_1$ optimization is by Cand\`{e}s et al. \cite{Candes
    Reweighted}. The main difference is that there is no prior
  information. At each step, the $\ell_1$ optimization is re-weighted
  using the estimate of the signal obtained in the last minimization
  step.} The second contribution is that we \emph{explicitly} compute the
relationship between  $p_i$, $w_{K_i}$,$\frac{n_i}{n}$, $1\leq i\leq u$ and
the number of measurements so that the unknown signal can be
recovered with overwhelming probability as $n\rightarrow\infty$ (the
so-called weak and strong thresholds) for measurement matrices drawn from an
i.i.d. Gaussian ensemble. The analysis uses the high-dimensional
geometry techniques first introduced by Donoho and Tanner
\cite{DT,D} (e.g., Grassmann angles) to obtain sharp thresholds for
compressed sensing. However, rather than the {\em
neighborliness} condition used in \cite{DT,D}, we find it more
convenient to use the null space characterization of Xu and Hassibi
\cite{Weiyu GM,StXuHa08}. The resulting Grassmannian manifold
approach is a general framework for incorporating additional factors
into compressed sensing: in \cite{Weiyu GM} it was used to
incorporate approximately sparse signals; here it is used to
incorporate prior information and weighted $\ell_1$ optimization.
Our analytic results allow us to precisely compute the optimal
weights for any $p_i$,$n_i$, $1\leq i\leq u$. We also provide certain robustness conditions for the recovery scheme for compressible signals or under model mismatch. We present
simulation results to show the advantages of the weighted method
over standard $\ell_1$ minimization. Furthermore, the results of this paper for the case of two classes ($u=2$)  builds a rigid framework for analyzing
certain classes of re-weighted $\ell_1$ minimization algorithms. In a re-weighted
$\ell_1$ minimization algorithm, the post processing information from
the estimate of the signal at each step can be viewed as additional
prior information about the signal, and can be incorporated into the
next step as appropriate weights. In a further work we have been
able to analytically prove the threshold improvement in a reweighted
$\ell_1$ minimization  using this framework \cite{KHXUAVHA Allerton}. It is worth mentioning that we have
prepared a software package based on the results of this paper for
threshold computation using weighted $\ell_1$ minimization, and it
is available in \cite{Softwate Package}.

The paper is organized as follows. In the next section we briefly
describe the notations that we use throughout the paper. In Section
\ref{sec:model} we  describe the model and state the principal
assumptions of nonuniform sparsity that we are interested in. We
also sketch the objectives that we are shooting for and, clarify
what we mean by \emph{recovery improvement} in the weighted $\ell_1$
case. In Section \ref{sec: Main results}, we skim through our
critical theorems and try to present the big picture of the main
results.  Section \ref{sec:Derivations} is dedicated to the concrete
derivation of these results. In Section \ref{sec:approximated supp recov}, we briefly introduce the reweighted $\ell_1$ minimization algorithm, and provide some insights in how the derivations of this work can be used to analyze the improved recovery thresholds. In Section \ref{sec:Simulation} some
simulation results are presented and are compared to the analytical
bounds of the previous sections. The paper ends with a conclusion
and discussion of future work in Section \ref{sec:Conclusion}.

\section{Basic Definitions and  Notations}

Throughout the paper, vectors are denoted by small boldface letters
$\x,\w,\z,\cdots$, scalars are shown by small regular letters
$a,b,\alpha,\cdots$, and matrices are denoted by bold capital
letters($\A,\bf{I},\cdots$). For referring to geometrical objects
and subspaces, we use Calligraphic notation, e.g.
$\mathcal{Z},\F,\G,\mathcal{P},\mathcal{C},\cdots$. This includes
the notations that we use to indicate the faces of a high
dimensional polytope, or the polytope itself. Sets and random
variables are denoted by regular capital letters($K,S,\cdots$). The
normal distribution with mean $\mu$ and variance $\sigma^2$ is
denoted by $\mathcal{N}(\mu,\sigma^2)$. For functions we use both
little and capital letters and it should be generally clear from the
context. We use the phrases RHS and LHS as abbreviations for Right
Hand Side and Left Hand Side respectively throughout the paper.

\begin{definition}
A random variable $Y$ is said to have a \emph{Half Normal}
distribution $HN(0,\sigma^2)$ if $Y = |X|$ where $X$ is a zero mean
normal variable $X\sim \mathcal{N}(0,\sigma^2)$.
\end{definition}

\section{Problem Description}
\label{sec:model}

We first define the signal model. For completeness, we present a
general definition.
\begin{definition}
\label{def:nonuniform sparse} Let $\mathcal{K} =
\{K_1,K_2,...,K_u\}$ be a partition of $\{1,2,\cdots,n\}$, i.e. ($K_i\cap K_j = \emptyset$ for $ i\neq j$, and
$\bigcup_{i=1}^u K_i = \{1,2,...,n\}$), and $P = \{p_1,p_2,\cdots,p_u\}$ be a set of
positive numbers in $[0,1]$. A $n\times 1$  vector $\x =
(x_1,x_2,\cdots,x_n)^T$ is said to be a \textbf{random nonuniformly
sparse} vector with sparsity fraction $p_{i}$ over the set $K_i$ for
$1\leq i\leq u$, if $\x$ is generated from the following random
procedure:

\begin{itemize}
\item Over each set $K_i$, $1 \leq i \leq u $, the set of nonzero entries of $\x$ is a random subset of size $p_i|K_i|$. In other words, a fraction $p_i$ of the entries are nonzero in $K_i$. $p_i$ is called the sparsity fraction over $K_i$. The values of the nonzero entries of $\x$ can arbitrarily be selected from any symmetric distribution. We can choose $\mathcal{N}(0,1)$ for simplicity.
 \end{itemize}
\end{definition}

In Figure \ref{fig:model}, a sample  nonuniformly sparse signal with
Gaussian distribution for nonzero entries is plotted. The number of
sets is considered to be $u=2$ and both classes have the same size
$\frac{n}{2}$, with $n=1000$. The sparsity fraction for the first class $K_1$ is $p_1
= 0.3$, and for the second class $K_2$ is $p_2=0.05$. In fact, the
signal is much sparser in the second half than it is in the first
half. The advantageous feature of this model is that all the
resulting computations are independent of the actual distribution on
the amplitude of the nonzero entries. However, as expected, it is
not independent of the properties of the measurement matrix. We
assume that the measurement matrix $\A$ is a $m\times n$ matrix with
i.i.d. standard Gaussian distributed $\mathcal{N}(0,1)$ entries,
with $\frac{m}{n}=\delta<1$. The measurement vector is denoted by
$\y$ and obeys the following:
\begin{equation}
\y = \A\x.
\end{equation}

\begin{figure}
  \centering
  \includegraphics[width=0.6\textwidth]{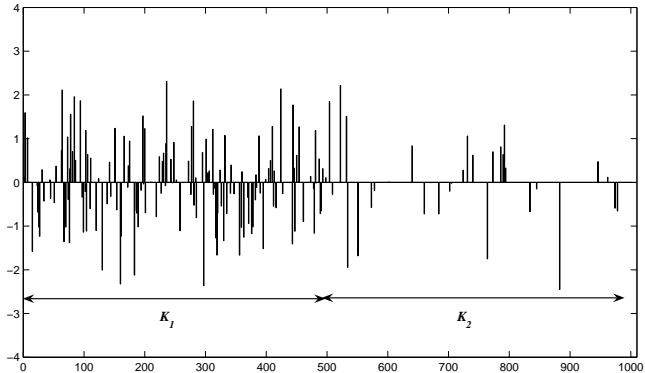}
   \caption{ \scriptsize Illustration of a nonuniformly sparse signal.}
  \label{fig:model}
\end{figure}

\noindent As mentioned in Section \ref{sec:Intro}, $\ell_1$
minimization can recover a randomly selected vector $\x$ with $k=\mu
n$ nonzero entries with high probability, provided $\mu$ is less
than a known function of $\delta$. $\ell_1$ minimization has the
following form:
\vspace*{-5pt}
\begin{equation}\label{eq:l_1}
\min_{\bf A\x=\y}{\|\x\|_1}.
\end{equation}

\noindent The reference ~\cite{DT} provides an explicit relationship
between $\mu$ and the minimum $\delta$ that guarantees success of
$\ell_1$ minimization recovery in the case of Gaussian measurements
and provides the corresponding numerical curve. The optimization in (\ref{eq:l_1}) is a
linear program and can be solved polynomially fast ($O(n^3)$).
However, it fails to encapsulate additional prior information of the
signal nature, might there be any such information available. One
can simply think of modifying (\ref{eq:l_1}) to a weighted $\ell_1$
minimization as follows:

\begin{equation} \label{eq:weighted l_1}
\min_{\bf A\x=\y}{\|\x\|_{\w,1}}=\min_{\A\x=\y}{\sum_{i=1}^{n}{w_i
|x_i|}}
\end{equation}

\noindent The index, $\w$, on the norm is an indication of the
$n\times 1$ positive weight vector. Now the questions are i) what is the
optimal set of weights for a certain set of available prior
information?, and ii) can one improve the recovery threshold using the
weighted $\ell_1$ minimization of (\ref{eq:weighted l_1}) by choosing a 
set of optimal weights? We have to be
more clear with our objective at this point and clarify what we
mean by improving the recovery threshold. Generally speaking, if a recovery method can
reconstruct all signals of a certain model with certainty, then that
method is said to be \emph{\textbf{strongly successful}} on that
signal model.  If we have a class of models that can be identified
with a parameter $\theta$, and if for all models corresponding to
$\theta < \theta_0$ a recovery scheme is strongly successful, then
the threshold $\theta_0$ is called a \emph{\textbf{strong recovery
threshold}} for the parameter $\theta$. For example, for fixed
$\frac{m}{n}$, if $k<n$ is sufficiently small, then $\ell_1$
minimization can provably recover all $k$-sparse signals, provided
that appropriate linear measurements have been made from the signal.
The maximum such $k$ is called the strong recovery threshold of the
sparsity for the success of $\ell_1$ minimization. Likewise, for a
fixed ratio $\mu=\frac{k}{n}$, the minimum ratio of measurements to
ambient dimension $\frac{m}{n}$ for which, $\ell_1$ minimization
always recovers $k$-sparse signals from the given $m$ linear
measurements is called the strong recovery threshold for the number of
measurements for $\ell_1$ minimization. In contrast, one can also
look into the \textbf{\emph{weak recovery}} threshold,
defined as the threshold below which, with very high probability a
random vector generated from the model is recoverable. For the nonuniformly sparse model, the quantity
of interest is the overall sparsity fraction of the model
defined as ($\frac{\sum_{i=1}^u p_in_i}{n}$). The question we ask is
whether  by adjusting $w_i$'s according to $p_i$'s one can extend
the strong or weak recovery threshold for sparsity fraction to a value
above the known threshold of $\ell_1$ minimization. Equivalently, for given classes $K_1,\cdots,K_u$ and sparsity fractions $p_i$'s, how much can the strong or weak threshold be improved for the minimum number of required measurements, as apposed to the case of uniform sparsity with the same overall sparsity fraction.

\section{Summary of Main Results}
\label{sec: Main results}

We state the two problems more formally using the notion of recovery thresholds
that we defined in the previous section.  We only consider the case
of $u=2$.
\begin{itemize}
\item{\textbf{Problem 1}}
Consider the random nonuniformly sparse model with two classes
$K_1,K_2$ of cardinalities $n_1=\gamma_1 n$ and $n_2= \gamma_2 n$
respectively, and given sparsity fractions $p_1$ and $p_2$. Let $\w$
be a given weight vector. As $n\rightarrow \infty$, what is the weak (strong)
recovery threshold for  $\delta = \frac{m}{n}$ so that a randomly
chosen vector (all vectors) $\x_0$ selected from the nonuniformly sparse model is successfully  recovered by the weighted
$\ell_1$ minimization of (\ref{eq:weighted l_1}) with high
probability?
\end{itemize}
Upon solving Problem.1, one can exhaustively search for the weight
vector $\w$ that results in the minimum recovery threshold for
$\delta$. This is what we recognize as the optimum set of weights.
So the second problem can be stated as:
\begin{itemize}
\item{\textbf{Problem 2}}
Consider the random nonuniformly sparse model defined by classes
$K_1,K_2$ of cardinalities $n_1$ and $n_2$ respectively, with
$\gamma_1= \frac{n_1}{n}$ and $\gamma_2= \frac{n_2}{n}$, and given
sparsity fractions $p_1$ and $p_2$. What is the optimum weight vector
$\w$ in (\ref{eq:weighted l_1}) that results in the minimum number
of measurements for almost sure recovery of signals generated from
the given random nonuniformly sparse model?
\end{itemize}

\noindent We will fully solve these problems in this paper. We first
connect the misdetection event to the properties of the measurement
matrix. For the non-weighted case, this has been considered in
\cite{StXuHa08} and is known as the null space property. We
generalize this result to the case of weighted $\ell_1$
minimization, and mention a necessary and sufficient condition for
(\ref{eq:weighted l_1}) to recover the original signal of interest.
The theorem is as follows
\begin{thm}\label{thm:Null space}
For all $n\times 1$ vectors $\x^*$ supported on the set
$K\subseteq\{1,2,...,n\}$, $\x^*$ is the unique solution to the
linear program $\min_{\A\x=\y}{\sum_{i=1}^{n}{w_i |x_i|}}$ with
$\y=\A\x^*$, if and only if for every vector
$\z=(z_1,z_2,\cdots,z_n)^T$ in the null space of $\A$, the following
holds: $\sum_{i\in K}{w_i|z_i|} \geq \sum_{i\in
\overline{K}}{w_i|z_i|}$.
\end{thm}

\noindent This theorem will be proved in Section
\ref{sec:Derivations}. As will be explained in Section \ref{sec:
Bounds on Pe}, Theorem \ref{thm:Null space} along with known facts
about the null space of random Gaussian matrices, help us interpret
the probability of recovery error in terms of a high dimensional
geometrical object called the complementary Grassmann angle; namely
the probability that a uniformly chosen $(n-m)$-dimensional subspace $\mathcal{Z}$
shifted by a point $\x$ of unity weighted $\ell_1$-norm,
$\sum_{i=1}^{n}w_ix_i=1$, intersects the weighted $\ell_1$-ball
$\mathcal{P}_{\w}=\{\y\in \mathbb{R}^n~|\sum_{i=1}^{n}w_i|y_i|\leq
1\}$ nontrivially at some other point besides $\x$. The shifted
subspace is denoted by ${\mathbf{\mathcal{Z}}} + \x$. The fact that
we can take for granted, without explicitly proving it, is that due
to the identical marginal distribution of the entries of $\x$ in
each of the sets $K_1$ and $K_2$, the entries of the optimal weight
vector take at most two (or in the general case $u$) distinct values
$w_{K_1}$ and $w_{K_2}$ depending on their index. In other words
\begin{equation}
\forall i\in\{1,2,\cdots,n\} ~~~w_i=\left\{\begin{array}{c}w_{K_1}
~if~i\in K_1\\ w_{K_2} ~if~i\in K_2
\end{array}\right.
\label{eq:w_i's}
\end{equation}

\noindent Leveraging on the existing techniques for computing the
complementary Grassmann angle~\cite{santalo,McMullen}, we will be able to state and prove the following
theorem along the same lines, which upper bounds the probability
that the weighted $\ell_1$ minimization does not recover the signal. Please note that in the following theorem, the rigorous mathematical definitions to some of the terms (internal angle and external angle) is not presented, due to the extent of descriptions. They will however be defined rigorously later in the derivations of the main results in Section \ref{sec:Derivations}.
\begin{thm}
\label{thm: bound on probability of error}
Let $K_1$ and $K_2$ be two disjoint subsets of $\{1,2,\cdots,n\}$
such that $|K_1|=n_1, |K_2|=n_2$, and $p_1$ and $p_2$ be real
numbers in $[0,1]$. Also, let $k_1 = p_1n_1$, $k_2 = p_2n_2$, and $E$ be the event that a random
nonuniformly sparse vector $\x_0$ (Definition \ref{def:nonuniform
sparse}) with sparsity fractions $p_1$ and $p_2$ over the sets $K_1$
and $K_2$ respectively is recovered via the weighted $\ell_1$
minimization of (\ref{eq:weighted l_1}) with $\y=\A\x_0$. Also, let
$E^c$ denote the complement event of $E$. Then
\small{
\begin{align}
\Prob\{E^c\} &\leq \sum_{{\tiny\begin{array}{c}0\leq t_1 \leq n_1-k_1\\0\leq t_2 \leq
n_2-k_2\\ t_1+t_2 > m-k_1-k_2+1\end{array}}} 2^{t_1+t_2+1}{n_1-
k_1\choose t_1}{n_2-k_2\choose t_2}
\beta(k_1,k_2|t_1,t_2)\zeta(t_1+k_1,t_2+k_2)
\label{eq:sumformula}
\end{align}
} \normalsize where $\beta(k_1,k_2|t_1,t_2)$ is the internal angle
between a  $(k_1+k_2-1)$-dimensional face $\mathcal{F}$ of the
weighted  $\ell_1$-ball
$\mathcal{P}_{\w}=\{y\in\mathbb{R}^{n}|\sum_{i=1}^{n}w_i|y_i|\leq
1\}$ with $k_1$ vertices supported on $K_1$ and $k_2$ vertices
supported on $K_2$, and another $(k_1+k_2+t_1+t_2-1)$-dimensional
face $\G$ that encompasses $\mathcal{F}$ and has $t_1+k_1$ vertices
supported on $K_1$ and the remaining $t_2+k_2$ vertices supported on
$K_2$.  $\zeta(d_1,d_2)$ is the external angle between a face $\G$
supported on set $L$ with $|L\cap K_1|=d_1$ and $|L\cap K_2|=d_2$
and the weighted $\ell_1$-ball $\mathcal{P}_{\w}$. See Section \ref{sec: Bounds on Pe} for the definitions of integral and external angles. 
\end{thm}

\noindent The proof of this theorem will be given in Section \ref{sec:special case of u=2}. We are interested in the regimes that make the
above upper bound decay to zero as $n\rightarrow\infty$, which
requires the cumulative exponent in (\ref{eq:sumformula}) to be
negative. We are able to calculate sharp upper bounds on the
exponents of the terms in (\ref{eq:sumformula}) by using large
deviations of sums of normal and half normal variables.  More
precisely, for small enough $\epsilon$, if we assume that the sum of
the terms corresponding to particular indices $t_1$ and $t_2$ in
(\ref{eq:sumformula}) is denoted by $F(t_1,t_2)$, and define $\tau_1
= \frac{t_1}{n}$ and $\tau_2 = \frac{t_2}{n}$, then we are able to
find and compute an exponent function
$\psi_{tot}(\tau_1,\tau_2)=\psi_{com}(\tau_1,\tau_2)-\psi_{int}(\tau_1,\tau_2)-\psi_{ext}(\tau_1,\tau_2)$
so that $\frac{1}{n}\log{F(t_1,t_2)} \sim \psi_{tot}(\tau_1,\tau_2)$
as $n\rightarrow \infty $. The terms $\psi_{com}(\cdot,\cdot)$,
$\psi_{int}(\cdot,\cdot)$ and $\psi_{ext}(\cdot,\cdot)$ are
contributions to the cumulative exponent $\psi_{tot}$ by the so
called combinatorial, internal angle and external angle terms
respectively, existing in the upper bound (\ref{eq:sumformula}).
The derivations of these terms will be elaborated in Section \ref{sec:exp cal.}. Consequently, we state a key theorem that is the
implicit answer to Problem 1.

\begin{thm}\label{thm:main delta bound}
Let $\delta = \frac{m}{n}$ be the ratio of the number of
measurements to the signal dimension, $\gamma_1 = \frac{n_1}{n}$ and
$\gamma_2 = \frac{n_1}{n}$. For fixed values of $\gamma_1$,
$\gamma_2$, $p_1$, $p_2$, $\omega = \frac{w_{K_2}}{w_{K_1}}$ ,
define $E$ to be the event that a random nonuniformly sparse vector
$\x_0$ (Definition \ref{def:nonuniform sparse}) with sparsity fractions
$p_1$ and $p_2$ over the sets $K_1$ and $K_2$ respectively with
$|K_1| = \gamma_1 n$ and $|K_2|=\gamma_2 n$ is recovered via the weighted $\ell_1$
minimization of (\ref{eq:l_1}) with $\y = \A\x_0$. There exists a
critical threshold
$\delta_c=\delta_c(\gamma_1,\gamma_2,p_1,p_2,\omega)$ such that if
$\delta=\frac{m}{n}\geq \delta_c$, then  $\Prob\{E^c\}$ decays
exponentially to zero as $n\rightarrow\infty$. Furthermore,
$\delta_c$ is given by \beq \addtolength{\fboxsep}{5pt} \boxed{
\begin{gathered}
\delta_c = \min\{\delta~|~\psi_{com}(\tau_1,\tau_2)-\psi_{int}(\tau_1,\tau_2)-\psi_{ext}(\tau_1,\tau_2)<0~ \forall ~0\leq \tau_1\leq \gamma_1(1-p_1), \nonumber \\
0\leq \tau_2\leq \gamma_2(1-p_2), \tau_1+\tau_2 >
\delta-\gamma_1p_1-\gamma_2p_2 \} \nonumber
\end{gathered}
} \eeq \noindent where $\psi_{com}$, $\psi_{int}$ and $\psi_{ext}$
are obtained from the following expressions:

\noindent Define $g(x)=\frac{2}{\sqrt{\pi}}e^{-{x^2}}$,
$G(x)=\frac{2}{\sqrt{\pi}}\int_{0}^{x}e^{-y^2}dy$ and let
$\varphi(.)$ and $\Phi(.)$ be the standard Gaussian pdf and cdf
functions respectively.
\begin{enumerate}
\item(Combinatorial exponent)
\beq \psi_{com}(\tau_1,\tau_2) =
\left(\gamma_1(1-p_1)H(\frac{\tau_1}{\gamma_1(1-p_1)})+\gamma_2(1-p_2)H(\frac{\tau_2}{\gamma_2(1-p_2)})+\tau_1+\tau_2\right)\log{2}
\label{eq:comb angle}
\eeq \noindent where $H(\cdot)$ is the entropy function defined by
$H(x) = -x\log{x}-(1-x)\log(1-x)$.

\item(External angle exponent)
Define $c=(\tau_1+\gamma_1p_1)+\omega ^2(\tau_2+\gamma_2p_2)$,
$\alpha_1=\gamma_1(1-p_1)-\tau_1$ and
$\alpha_2=\gamma_2(1-p_2)-\tau_2$. Let $x_0$ be the unique solution
to $x$ of the following:
\begin{equation*}
2c-\frac{g(x)\alpha_1}{xG(x)}-\frac{\omega g(\omega
x)\alpha_2}{xG(\omega x)}=0
\end{equation*}
Then
\begin{equation}
\psi_{ext}(\tau_1,\tau_2) =
cx_0^2-\alpha_1\log{G(x_0)}-\alpha_2\log{G(\omega x_0)}
\label{eq:ext angle}
\end{equation}

\item(Internal angle exponent)
Let $b=\frac{\tau_1+\omega ^2\tau_2}{\tau_1+\tau_2}$,
$\Omega'=\gamma_1p_1+\omega ^2\gamma_2p_2$ and
$Q(s)=\frac{\tau_1\varphi(s)}{(\tau_1+\tau_2)\Phi(s)}+\frac{\omega
\tau_2\varphi(\omega s)}{(\tau_1+\tau_2)\Phi(\omega s)}$. Define the
function $\hat{M}(s)=-\frac{s}{Q(s)}$ and solve for $s$ in
$\hat{M}(s)=\frac{\tau_1+\tau_2}{(\tau_1+\tau_2)b+\Omega'}$. Let the
unique solution be $s^*$ and set $y=s^*(b-\frac{1}{\hat{M}(s^*)})$.
Compute the rate function $\Lambda^*(y)= sy
-\frac{\tau_1}{\tau_1+\tau_2}\Lambda_1(s)-\frac{\tau_2}{\tau_1+\tau_2}\Lambda_1(\omega
s)$ at the point $s=s^*$, where $\Lambda_1(s) = \frac{s^2}{2}
+\log(2\Phi(s))$. The internal angle exponent is then given by:
\begin{equation}
\psi_{int}(\tau_1,\tau_2) =
(\Lambda^*(y)+\frac{\tau_1+\tau_2}{2\Omega'}y^2+\log2)(\tau_1+\tau_2)
\label{eq:int angle}
\end{equation}
\end{enumerate}
\end{thm}

\noindent Theorem \ref{thm:main delta bound} is a powerful result,
since it allows us to find (numerically) the optimal set of weights
for which the fewest possible measurements are needed to recover the
signals almost surely. To this end, for
fixed values of $\gamma_1$, $\gamma_2$, $p_1$ and $p_2$, one should find the ratio
$\frac{w_{K_2}}{w_{K_1}}$ for which the critical threshold
$\delta_c(\gamma_1,\gamma_2,p_1,p_2,\frac{w_{K_2}}{w_{K_1}})$ from
Theorem \ref{thm:main delta bound} is minimum. We discuss this by
some examples in Section \ref{sec:Simulation}. A generalization of theorem \ref{thm:main delta bound} for a nonuniform model with an arbitrary number of classes ($u\geq 2$) will be given in Section \ref{sec:generalizations}. 



As mentioned earlier, using Theorem \ref{thm:main delta bound}, it
is possible to find the optimal ratio
$\frac{w_{K_2}}{w_{K_1}}$. It however requires an exhaustive search over the $\delta_c$ threshold for all possible values of $\omega$. For $\gamma_1=\gamma_2=0.5$, $p_1=0.3$
and $p_2=0.05$, we have numerically computed
$\delta_c(\gamma_1,\gamma_2,p_1,p_2,\frac{w_{K_2}}{w_{K_1}})$ as a
function of $\frac{w_{K_2}}{w_{K_1}}$ and depicted the resulting
curve in Figure \ref{fig:critical delta}. This suggests that
$\frac{w_{K_2}}{w_{K_1}}\approx2.5$ is the optimal ratio that one
can choose. Later we will confirm this using simulations.

\begin{figure}
  \centering
  \subfloat[$p_1=0.4$, $p_2=0.05$.]{\label{fig:critical delta}\includegraphics[width=0.46\textwidth]{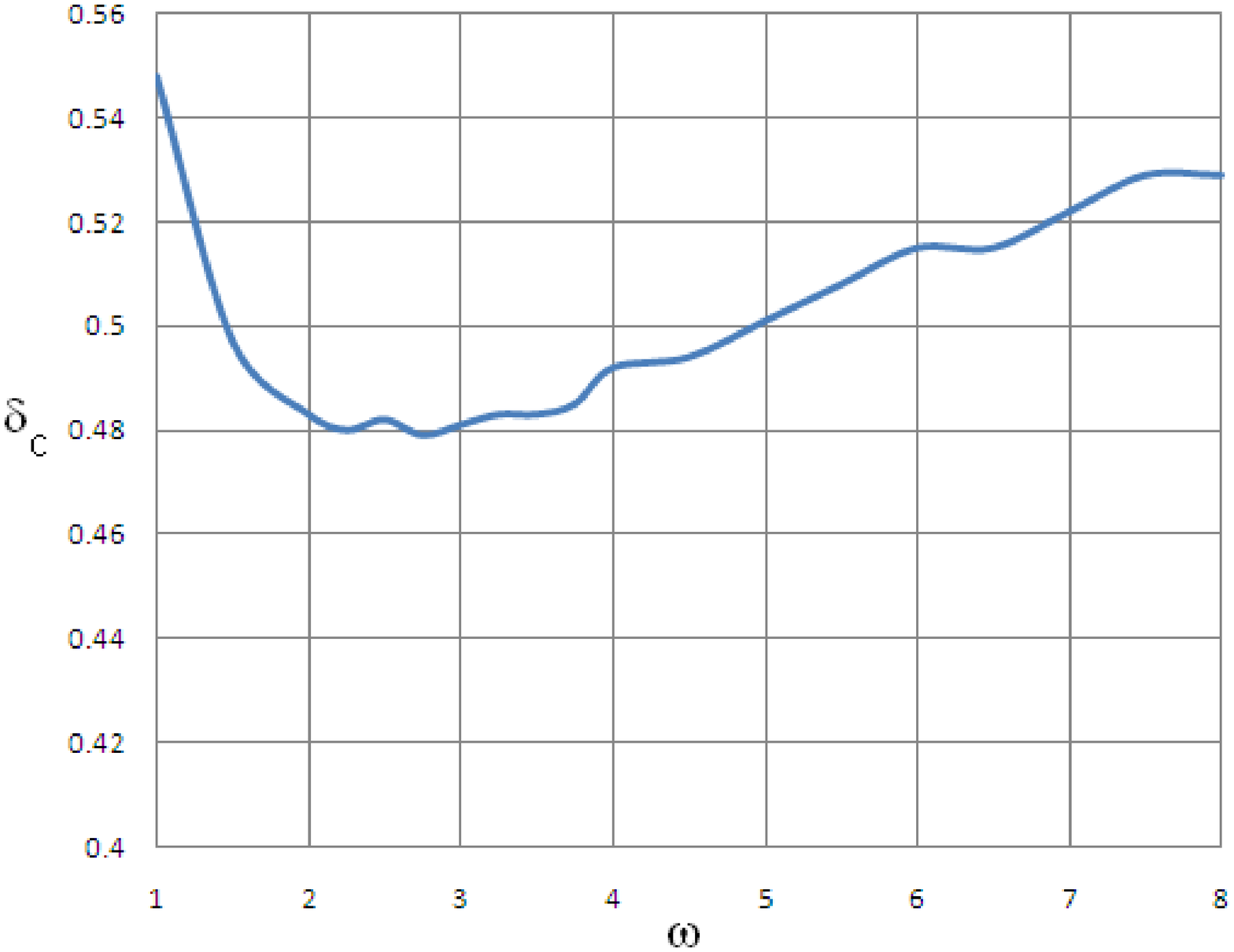}}
  \subfloat[$p_1=0.65$, $p_2=0.1$.]{\label{fig:critical delta other setting}\includegraphics[width=0.46\textwidth]{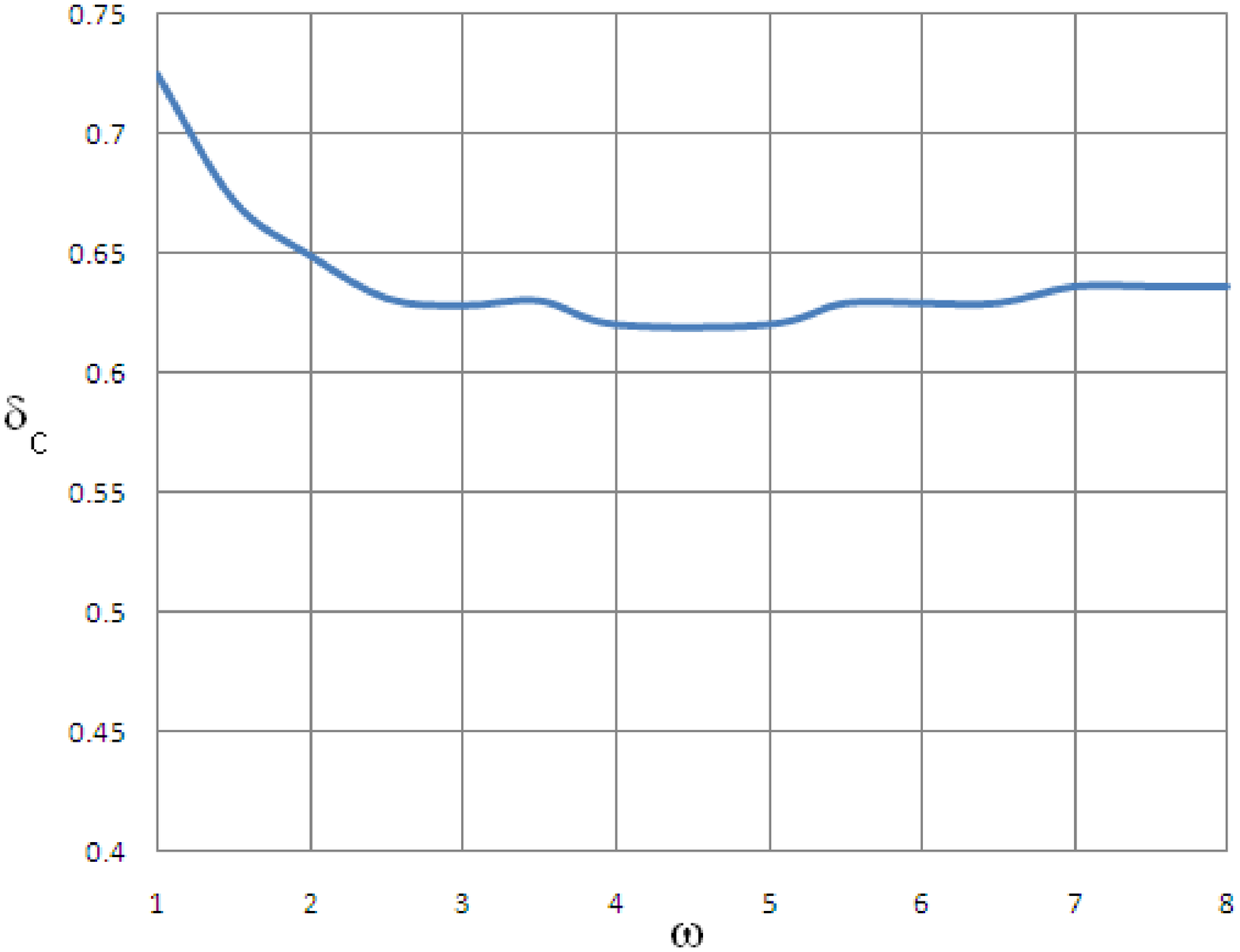}}
  \caption{\scriptsize  $\delta_c$ as a function of $\omega = \frac{w_{K_2}}{w_{K_1}}$ for $\gamma_1=\gamma_2=0.5$.}
  \label{fig:delta_c}
\end{figure}

Note that $\delta_c$ given in Theorem \ref{thm:main delta bound} is a weak bound on the ratio $\delta = \frac{m}{n}$. In other words, it determines the minimum number of measurements so that for a random sparse signal from the nonuniform sparse model and a random support set, the recovery is successful with high probability. It is possible to obtain a strong bound for $\delta$, using a union bound on all possible support sets in the model, and all possible sign patterns of the sparse vector. Similarly, a sectional bound can be defined which accounts for all possible support sets but almost all sign patterns.  Therefore, the expressions for the strong and sectional thresholds, which we denote by $\delta_c^{(S)}$ and $\delta_c^{(T)}$ are very similar to $\delta_c$ in Theorem \ref{thm:main delta bound}, except for a slight modification in the combinatorial exponent term $\psi_{com}$. This will be elaborated in Section \ref{sec:exp cal.}.

It is worthwhile to consider some asymptotic cases of the presented nonuniform model and some of their implications. First of all, when one of the subclasses is empty, e.g. $\gamma_1 = 0$, then the obtained weak and strong thresholds are  equal to the corresponding thresholds of $\ell_1$ minimization for a sparsity fraction $p = p_2$. Furthermore, if the sparsity fractions $p_1$ and $p_2$ over the two classes are equal, and a unitary weight $\omega = 1$ is used, then the weak threshold $\delta_c$ is equal to the threshold of $\ell_1$ minimization for a sparsity fraction $p=p_1=p_2$. In other words:
\beq \delta_c(\gamma_1,\gamma_2,p,p,1) = \delta_c(0,1,0,p,1).\eeq
\noindent This follows immediately from the derivations of the exponents in Theorem \ref{thm:main delta bound}. However, the latter is not necessarily true for the strong threshold. In fact the computation of the strong threshold for regular $\ell_1$ minimization involves a union bound over a larger set of possible supports, and therefore the combinatorial exponent becomes larger. Therefore:
\beq \delta_c^{(S)}(\gamma_1,\gamma_2,p,p,1) \leq \delta_c^{(S)}(0,1,0,p,1).\eeq

A very important asymptotic case is when the unknown signal is fully dense over one of the subclasses, e.g. $p_1 = 1$, which accounts for a \emph{partially known} support. This model is considered in the work of Vaswani et al.~\cite{Vas}, with the motivation that in some applications (or due to previous processing steps), part of the support set can be fully identified\footnote{Thanks to anonymous reviewers for pointing this out to us!}. If the dense subclass is $K_1$ and $K_2 = K_1^c$, then  \cite{Vas} suggests solving the following minimization program:
\beq \min_{\A\x = \y} \|\x_{K_2}\|_1 \label{eq:l1 over K_2}.\eeq
\noindent It is possible to find exact thresholds for the above problem using the weighted $\ell_1$ minimization machinery presented in this paper. First, note that (\ref{eq:l1 over K_2}) is the asymptotic solution of the following weighted $\ell_1$ minimization, when  $\omega\rightarrow\infty$
\beq \min_{\A\x = \y} \|\x_{K_1}\|_1+\omega\|\x_{K_2}\|_1 \label{eq:l1 over K_1 and K_2}.\eeq
Therefore the recovery threshold for (\ref{eq:l1 over K_2}) can be given by $\delta_c(\gamma_1,\gamma_2,1,p_2,\omega)$ for $\omega\rightarrow\infty$. We prove the following theorem about the latter threshold:
\begin{thm}
\label{thm:p1=1 theorem}
If $\omega\rightarrow\infty$, then $\delta_c(\gamma_1,\gamma_2,1,p_2,\omega)\rightarrow \gamma_1 + \gamma_2\delta_c(0,1,0,p_2,1)$. In other words, when a subset of entries of size $\gamma_1n$ are known to be nonzero, the minimum number of measurements that is required for almost surely successful recovery using (\ref{eq:l1 over K_2}) is equal to the total number of measurements needed if we were allowed to independently make measurements from the two parts and recover each using $\ell_1$ minimization.
\end{thm}
\noindent The proof of this theorem is given in Appendix \ref{App: proof of theorem p1=1}.

A very important factor regarding the performance of any recovery method is its robustness. In other words, it is important to understand how resilient the recovery is in the case of compressible signals or in the presence of noise or model mismatch(i.e. incorrect knowledge of the the sets or sparsity factors). We address this in the following theorem.  
\begin{thm}
\label{thm:robustness}
Let $K_1$ and $K_2$ be two disjoint subsets of $\{1,2,\cdots,n\}$, with $|K_1| = \gamma_1n$,$|K_2| = \gamma_2 n$ and $\gamma_1 + \gamma_2 = 1$. Also suppose that the dimensions of the measurement matrix $\A$ satisfy $\delta =\frac{m}{n}\geq \delta_c^{(S)}(\gamma_1,\gamma_2,p_1,p_2,\omega)$ for positive real numbers $p_1$ and $p_2$ in $[0,1]$ and $\omega >0$.
For  positive $\epsilon_1,\epsilon_2$, assume that $L_1$ and $L_2$ are arbitrary subsets of $K_1$ and $K_2$ with cardinalities $(1-\epsilon_1)\gamma_1p_1n$ and $(1-\epsilon_2)\gamma_2p_2n$ respectively.
With high probability, for every vector $\x_0$, if $\hat{\x}$ is the solution to the following linear program:
\bea \min_{\A\x = \A\x_0} \|\x_{K_1}\|_1 + \omega\|\x_{K_2}\|_1.\eea
\noindent Then the following holds
\beq \|(\x_0-\hat{\x})_{K_1}\|_1 + \omega\|(\x_0-\hat{\x})_{K_2}\|_1 \leq C_{\epsilon_1,\epsilon_2} \left( \|(\x_0)_{\overline{L_1}\cap K_1}\|_1+\omega\|(\x_0)_{\overline{L_2}\cap K_2}\|_1\right),
\label{eq:robustness}
\eeq
\noindent where
\beq \nonumber  C_{\epsilon_1,\epsilon_2}  = \frac{1 + \min(\frac{\epsilon_1p_1}{1-p_1},\frac{\epsilon_2p_2}{1-p_2})}{1 -   \min(\frac{\epsilon_1p_1}{1-p_1},\frac{\epsilon_2p_2}{1-p_2})} .\eeq
\end{thm}
The above theorem has the following implications. First, if $\x_0$ is a (compressible) vector, such that its ``significant'' entries follow a nonuniform sparse model, then the recovery error of the corresponding weighted  $\ell_1$ minimization can be bounded in terms of the $\ell_1$ norm of the ``insignificant'' part of $\x_0$(i.e. the part where a negligible fraction of the energy of the signal is located or most entries have significantly small values, compared to the other part that has an overall large norm). Theorem \ref{thm:robustness} can also be interpreted as the robustness of weighted $\ell_1$ scheme to the model mismatch. If $K_1,K_2,p_1,p_2$ are the estimates of an actual nonuniform decomposition for $\x_0$ (based on which the minimum number of required measurements have been estimated), then the recovery error can be relatively small if the model estimation error is slight. Theorem \ref{thm:robustness} will be proved in Section \ref{sec:robustness}.
\section{Derivation of the main results}
\label{sec:Derivations} In this section we provide detailed proofs
to the claims of Section \ref{sec: Main results}. Let $\x_0$ be a random nonuniformly sparse signal with sparsity
fractions $p_1$ and $p_2$ over the index subsets $K_1$ and $K_2$
respectively (Definition \ref{def:nonuniform sparse}), and let
$|K_1|=n_1$ and $|K_2|=n_2$ . Also let $K$ be the support of $\x$. Let $E$ be the event that $\x$ is
recovered exactly by (\ref{eq:weighted l_1}), and $E^c$ be its
complimentary event. In order to bound the conditional error probability
$\Prob\{E^c\}$ we adopt the idea
of~\cite{StXuHa08} to interpret the failure recovery event ($E^c$)
in terms of the null space of the measurement matrix $\A$. This is stated in Theorem \ref{thm:Null space}, which we prove here.

\begin{proof}[proof of Theorem \ref{thm:Null space}]
Suppose the mentioned null space condition holds and define $\hat{\x} = argmin_{\A\x=\y}{\sum_{i=1}^{n}{w_i |x_i|}}$. Let $\W = diag(w_1,w_2,\cdots,w_n)$. By triangular inequality, we have:

\bea
\nonumber \|\W\hat{\x}\|_1  = \|(\W\hat{\x})_K\|_1 +\|(\W\hat{\x})_{\overline{K}}\|_1 &=& \|(\W\x^*+\W\hat{\x}-\W\x^*)_K\|_1 +\|(\W\hat{\x})_{\overline{K}}\|_1  \\
\nonumber &\geq& \|(\W\x^*)_K\|_1 -\|(\W\hat{\x}-\W\x^*)_K\|_1 +
\|(\W\hat{\x}-\W\x^*)_{\overline{K}}\|_1  \\ \nonumber&\geq&
\|\W\x^*\|_1 \eea Where the last inequality is a result of the fact
that $\hat{\x}-\x^*$ is in the null space of $\A$ and satisfies the
mentioned null space condition. However, by assumption if
$\hat{\x}\neq \x^*$ then $\|\W\hat{\x}\|_1 \leq \|\W\x^*\|_1 $. This
implies that $\hat{\x}=\x^*$. Conversely, suppose there is some
vector $\z$ in $\mathcal{N}(\A)$ such that $\|(\W\z)_K\|_1 >
\|(\W\z)_{\overline{K}}\|_1$. Taking define $\x^* = (\z_K ~0)^T$ and
$\hat{\x} = (0 ~\z_{\overline{K}})^T$ implies that $\A\x^* = \A\hat{\x}$
and $\|\W\x^*\|_1 > \|\W\hat{\x}\|_1$. Therefore, $\x^*$ cannot be
recovered from the weighted $\ell_1$ minimization.
\end{proof}

From this point on, we follow closely the steps towards calculating
the upper bound on the failure probability from \cite{Weiyu GM}, but
with appropriate modifications. The key to our derivations is the
following lemma which will be proven in Appendix \ref{App: proof of
lemma baselemma}.

\begin{lem}
For a certain subset $K \subseteq \{1,2,...,n\}$ with $|K|=k$, the
event that the null-space $\mathcal{N}(A)$ satisfies
\begin{equation}\label{eq:iffcon}
\sum_{i\in K}w_i|z_i|\leq \sum_{i\in\overline{K}}w_i|z_i|, \forall
\z \in \mathcal{N}(A),
\end{equation}
is equivalent to the event that for each $\x$ supported on the set
$K$ (or a subset of $K$)
\begin{equation}
\sum_{i\in K} w_i|x_i+z_i|+ \sum_{i\in\overline{K}}w_i|z_i| \geq
\sum_{i\in K} w_i|x_i| , \forall \z \in \mathcal{N}(A).
\label{eq:xcondition}
\end{equation}
\label{lemma:baselemma}
\end{lem}

\subsection{Upper Bound on the Failure Probability}
\label{sec: Bounds on Pe} Knowing Lemma \ref{lemma:baselemma}, we
are now in a position to derive the probability that condition
(\ref{eq:iffcon}) holds for a support set $K$ with $|K|=k$, if we
randomly choose an i.i.d. Gaussian matrix $\A$. In the case of a
random i.i.d. Gaussian matrix, the distribution of null space of
$\A$ is right-rotationally invariant, and sampling from this
distribution is equivalent to uniformly sampling a random
$(n-m)$-dimensional subspace $\mathcal{Z}$ from the Grassmann
manifold $\text{Gr}_{(n-m)}(n)$. The Grassmann manifold
$\text{Gr}_{(n-m)}(n)$ is defined as the set of all
$(n-m)$-dimensional subspaces of $\mathbb{R}^n$. We need to upper bound the complementary probability
$P=\Prob\{E^c\}$, namely the
probability that the (random) support set $K$ of $\x$ (of random sign pattern) fails the null
space condition (\ref{eq:xcondition}). We denote the null space of
$\A$ by $\mathcal{Z}$. Because $\mathcal{Z}$ is a linear space, for
every vector $\z\in \mathcal{Z}$,  $\alpha\z$ is also in
$\mathcal{Z}$ for all $\alpha\in \mathbb{R}$. Therefore, if for a
$\z\in\mathcal{Z}$ and $\x$ condition (\ref{eq:xcondition}) fails,
by a simple re-scaling of the vectors, we may assume without loss of
generality that $\x$ lies on the surface of any convex ball that
surrounds the origin. Therefore we restrict our attention to those
vectors $\x$ from the weighted $\ell_1$-sphere:
\begin{equation*}
\{\x\in \mathbb{R}^n ~|~ \sum_{i=1}^{n}w_i|x_i|=1 \}
\end{equation*}
that are only supported on the set $K$ , or a subset of it. Since we are assuming that the distribution of the nonzero entries of $\x$ is symmetric, we can write:
\vspace*{-7pt}
\begin{equation}
P = P_{K,-} \label{eq:error bound 2}
\end{equation}

\noindent where $P_{K,-}$ is the probability that for a specific
\emph{support set} $K$ , there
exist a $k$-sparse vector $\x$ of a specific \emph{sign pattern}
which  fails the condition (\ref{eq:xcondition}). By symmetry,
without loss of generality, we assume the signs of the elements of
$\x$ to be non-positive. Now we can focus on deriving the probability $P_{K,-}$. Since $\x$
is a non-positive $k$-sparse vector supported on the set $K$  and
can be restricted to the weighted $\ell_1$-sphere $\{\x\in
\mathbb{R}^n ~|~ \sum_{i=1}^{n}w_i|x_i|=1 \}$, $\x$ is also on a
$(k-1)$-dimensional face, denoted by $\mathcal{F}$, of the weighted
$\ell_1$-ball $\mathcal{P}_{\w}$:

\begin{equation}
\mathcal{P}_{\w}=\{\y\in \mathbb{R}^n~|~\sum_{i=1}^{n}w_i|y_i| \leq
1\}
\end{equation}
\begin{figure}[t]
\centering
   \subfloat[]{\label{fig:l_1 ball}\includegraphics[width=0.5\textwidth]{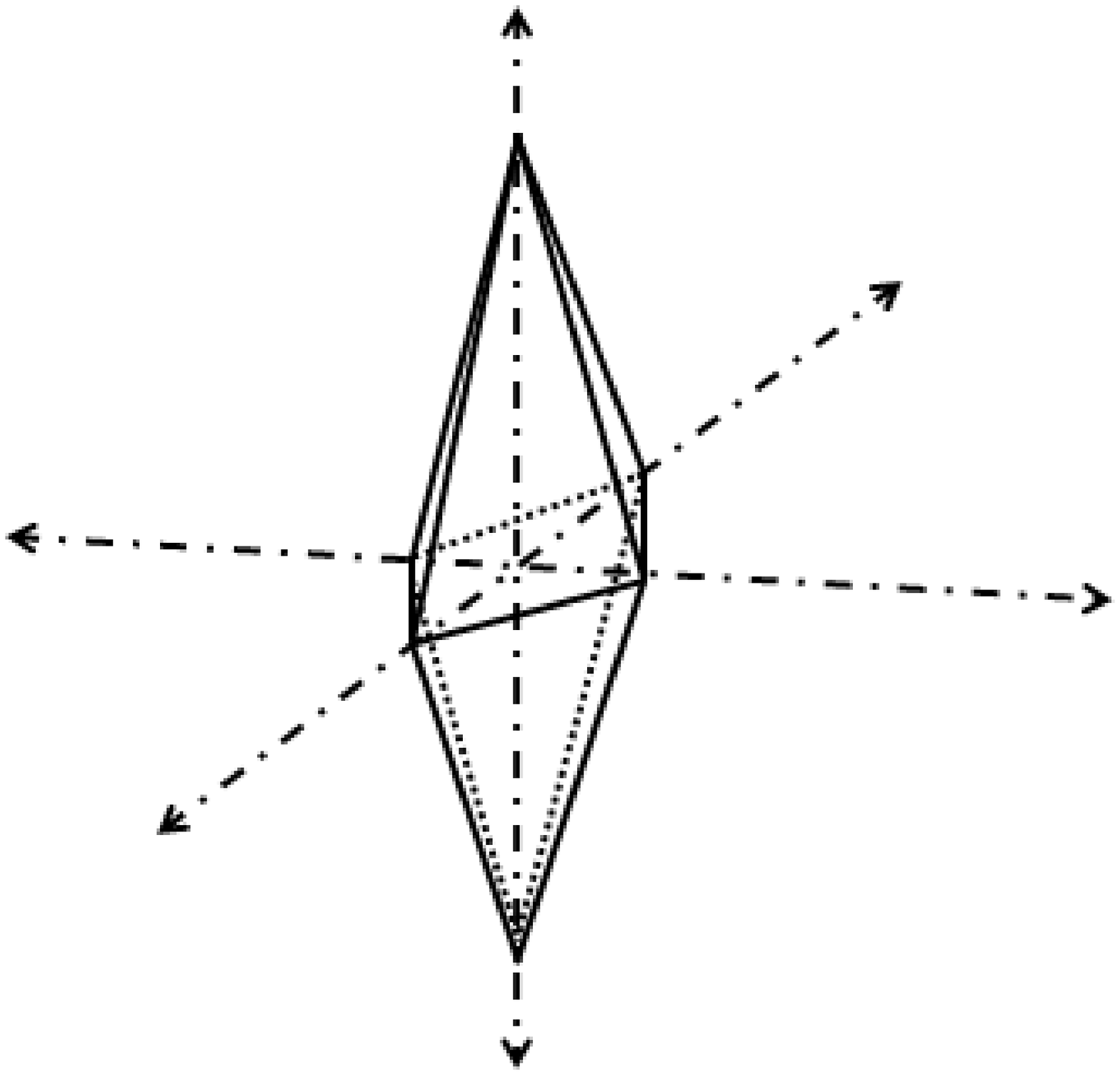}}
   \subfloat[]{\label{fig:l_1 ball
intersects}\includegraphics[width=0.5\textwidth]{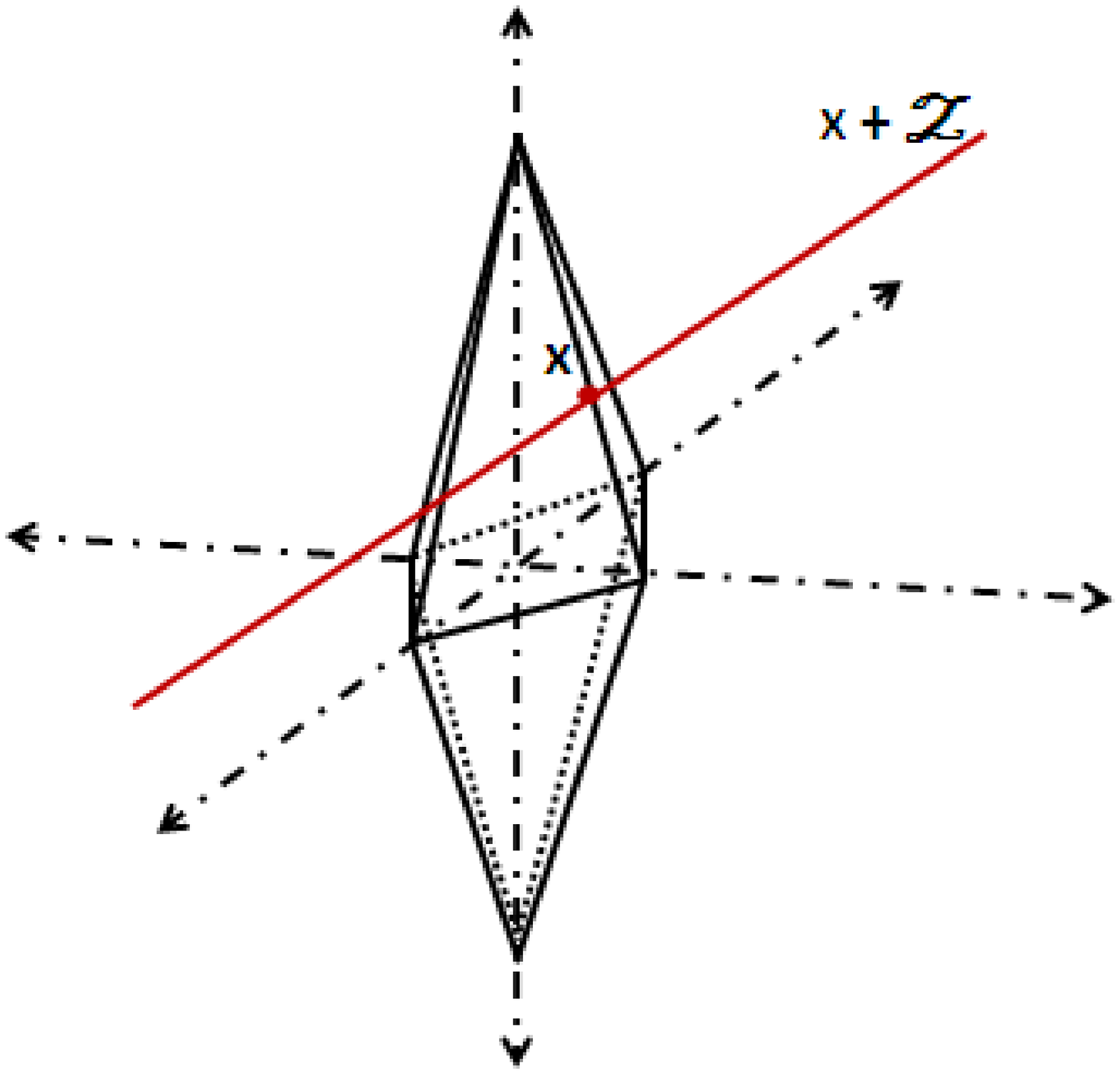}}
\caption{\scriptsize A weighted $\ell_1$-ball, $\mathcal{P}_{\w}$,
in $\mathbb{R}^3$ (a), and  a linear hyperplane $\mathcal{Z}$
passing through a point $\x$ in the interior of a one dimensional
face of $\mathcal{P}_{\w}$ (b).}
\end{figure}
\noindent The subscript $\w$ in $\mathcal{P}_{\w}$ is an indication
of the weight vector $\w=(w_1,w_2,\cdots,w_n)^T$. Figure
\ref{fig:l_1 ball} shows $\mathcal{P}_{\w}$ in $\mathbb{R}^3$ for
some nontrivial weight vector $\w$. Now the probability $P_{K,-}$ is
equal to the probability that there exists an $\x \in \F$, and there
exists a $ \z \in \mathcal{Z}$ ($\z\neq 0$) such that
\begin{equation} \sum_{i\in K}w_i|x_i+z_i|+ \sum_{i\in\bar{K}}w_i|z_i| \leq
\sum_{i\in K}w_i|x_i|=1.
\end{equation}

%
We start by studying the case for a specific point $\x \in \F$ and,
without loss of generality, we assume $\x$ is in the relative
interior of this $(k-1)$-dimensional face $\mathcal{F}$. For this
particular $\x$ on $\mathcal{F}$, the probability, denoted by
$P_{\x}'$, that there exists a $\z \in \mathcal{Z}$ ($\z\neq 0$)
such that
\begin{equation} \sum_{i\in K}w_i|x_i+z_i|+ \sum_{i\in\bar{K}}w_i|z_i| \leq
\sum_{i\in K}w_i|x_i|=1.
\end{equation}
is essentially the probability that a uniformly chosen
$(n-m)$-dimensional subspace $\mathcal{Z}$ shifted by the point
$\x$, namely $(\mathcal{Z}+\x)$, intersects the weighted
$\ell_1$-ball $\mathcal{P}_{\w}$ \emph{non-trivially}, namely, at
some other point besides $\x$ (Figure \ref{fig:l_1 ball
intersects}). From the fact that $\mathcal{Z}$ is a \emph{linear}
subspace, the event that $(\mathcal{Z}+\x)$ intersects
$\mathcal{P}_{\w}$ is equivalent to the event that $\mathcal{Z}$
intersects nontrivially with the cone $\mathcal{C}_{\w}(\x)$
obtained by observing the weighted $\ell_1$-ball $\mathcal{P}_{\w}$
from the point $\x$. (Namely,  $\mathcal{C}_{\w}(\x)$ is conic hull
of the point set $(\mathcal{P}_{\w}-\x)$ and of course
$\mathcal{C}_{\w}(\x)$ has the origin of the coordinate system as
its apex.) However, as noticed in the geometry for convex polytopes
\cite{Grunbaumpaper,Grunbaumbook}, the cones
 $\mathcal{C}_{\w}(\x)$ are identical for any $\x$ lying in the relative
interior of the face $\mathcal{F}$.  This means that the probability
$P_{K,-}$ is equal to $P_{\x}'$, regardless of the fact that $\x$ is
only a single point in the relative interior of the face
$\mathcal{F}$. There are some singularities here because $\x \in \F$
may not be in the relative interior of $\mathcal{F}$, but it turns
out that the $\mathcal{C}_{\w}(\x)$ in this case is only a subset of
the cone we get when $\x$ is in the relative interior of
$\mathcal{F}$. So we do not lose anything if we restrict $\x$ to be
in the relative interior of the face $\mathcal{F}$, namely we have
\begin{equation*}
P_{K,-}=P_{\x}'.
\end{equation*}

\noindent Now we only need to determine $P_{\x}'$. From its
definition, $P_{\x}'$ is exactly the \emph{\textbf{complementary
Grassmann angle}} \cite{Grunbaumpaper} for the face $\mathcal{F}$
with respect to the polytope $\mathcal{P}_{\w}$ under the Grassmann
manifold $\text{Gr}_{(n-m)}(n)$:
a uniformly distributed $(n-m)$-dimensional subspace $\mathcal{Z}$
from the Grassmannian manifold $\text{Gr}_{(n-m)}(n)$ intersecting
non-trivially with the cone $\mathcal{C}_{\w}(\x)$ formed by
observing the weighted $\ell_1$-ball $\mathcal{P}_{\w}$ from the
relative interior point $\x \in \F$.

Building on the works by L.A. Santal\"{o} \cite{santalo} and P.
McMullen \cite{McMullen} in high dimensional geometry and
convex polytopes, the complementary Grassmann angle for the
$(k-1)$-dimensional face $\mathcal{F}$ can be explicitly expressed
as the sum of products of internal angles and external angles
\cite{Grunbaumbook}:
\begin{equation}
2\times \sum_{s \geq 0}\sum_{G \in \Im_{m+1+2s}(\mathcal{P}_{\w})}
{\beta(\F,\G)\zeta(\G,\mathcal{P}_{\w})}, \label{eq:angformula}
\end{equation}
where $s$ is any nonnegative integer, $\G$ is any
$(m+1+2s)$-dimensional face of the $\mathcal{P}_{\w}$
($\Im_{m+1+2s}(\mathcal{P}_{\w})$ is the set of all such faces),
$\beta(\cdot,\cdot)$ stands for the internal angle and
$\zeta(\cdot,\cdot)$ stands for the external angle, and are defined as follows
\cite{Grunbaumbook,McMullen}:
\begin{itemize}
\item An internal angle $\beta(\F_1, \F_2)$ is the fraction of the
hypersphere $S$ covered by the cone obtained by observing the face
$\mathcal{F}_2$ from the face $\mathcal{F}_1$. \footnote{Note the
dimension of the hypersphere $S$ here matches the dimension of the
corresponding cone discussed. Also, the center of the hypersphere is
the apex of the corresponding cone. All these defaults also apply to
the definition of the external angles. } The internal angle
$\beta(\F_1, \F_2)$ is defined to be zero when
$\mathcal{F}_1\nsubseteq \mathcal{F}_2$ and is defined to be one if
$\mathcal{F}_1=\F_2$.

\item An external angle $\zeta(\mathcal{F}_3, \mathcal{F}_4)$ is the fraction of the
hypersphere $S$ covered by the cone of outward normals to the
hyperplanes supporting the face $\mathcal{F}_4$ at the face
$\mathcal{F}_3$. The external angle $\zeta(\mathcal{F}_3,
\mathcal{F}_4)$ is defined to be zero when $\mathcal{F}_3 \nsubseteq
\mathcal{F}_4$ and is defined to be one if
$\mathcal{F}_3=\mathcal{F}_4$.
\end{itemize}

In order to calculate the internal and external angles, it is
important to use the symmetrical properties of the weighted
cross-polytope $\mathcal{P}_{\w}$. First of all, $\mathcal{P}_{\w}$
is nothing but the convex hull of the following set of $2n$ vertices
in $\mathbb{R}^n$
\begin{equation}
\mathcal{P}_{\w}=conv\{\pm \frac{\e_i}{w_i}~|~1\leq i \leq n\}
\end{equation}
\noindent where $\e_i~1\leq i\leq n$ is the standard unit vector in
$\mathbb{R}^n$ with the $i$th entry equal to $1$. Every
$(k-1)$-dimensional face $\mathcal{F}$ of $\mathcal{P}_{\w}$ is
simply the convex hull of $k$ of the linearly independent vertices
of $\mathcal{P}_{\w}$. In that case we  say that $\mathcal{F}$
\emph{is supported} on the index set $K$ of the $k$ indices
corresponding to the nonzero coordinates of the vertices of $\F$  in
$\mathbb{R}^n$. More precisely, if
$\mathcal{F}=conv\{j_1\frac{\e_{i_1}}{w_{i_1}},j_2\frac{\e_{i_2}}{w_{i_2}},\cdots,j_k\frac{\e_{i_k}}{w_{i_k}}\}$
with $j_i\in \{-1,+1\}~\forall 1\leq i\leq k$, then $\F$ is said to
be supported on the set $K=\{i_1,i_2,\cdots,i_k\}$.

\subsection{Special Case of $u=2$}
\label{sec:special case of u=2}
The derivations of the previous section were for a general
weight vector $\w$. We now restrict ourselves to the case of two
classes, i.e. $u=2$, namely $K_1$ and $K_2$ with $|K_1|=n_1$ and
$|K_2|=n_2$.  For this case, we may assume that $w_i's$ have the
following particular form
\begin{equation}
\forall i\in\{1,2,\cdots,n\} ~~~w_i=\left\{\begin{array}{c}w_{K_1}
~if~i\in K_1\\ w_{K_2} ~if~i\in K_2
\end{array}\right.
\label{eq:w_i's 1}
\end{equation}

\begin{proof}[proof of Theorem \ref{thm: bound on probability of error}] The choice of $\w$ as in (\ref{eq:w_i's 1}) results in $\mathcal{P}_{\w}$  having two
classes of geometrically identical vertices, and many of faces of $\mathcal{P}_{\w}$
being isomorphic. In fact, two faces $\mathcal{F}$ and
$\mathcal{F}'$ of $\mathcal{P}_{\w}$ that are respectively supported
on the sets $K$ and $K'$ are geometrically isomorphic \footnote{This
means that there exists a rotation matrix $\Theta\in
\mathbb{R}^{n\times n}$ which is unitary i.e.  $\Theta^T\Theta = I$,
and maps $\F$ isometrically to $\F'$ i.e. $\F'=\Theta \F$.} if
$|K\cap K_1| = |K'\cap K_1|  $ and $|K\cap K_2| = |K'\cap
K_2|$\footnote{Remember that $K_1$ and $K_2$ are the same sets as
defined in the model description of Section \ref{sec:model}.}. In
other words the only thing that distinguishes the morphology of the
faces of $\mathcal{P}_{\w}$ is the proportion of their support sets
that is located in $K_1$ or $K_2$. Therefore for two faces $\F$ and
$\G$ with $\F$ supported on $K$ and $\G$ supported on $L$
($K\subseteq L$), $\beta(\mathcal{F},\mathcal{G})$ is only a
function of the parameters $k_1=|K\cap K_1|$, $k_2=|K\cap K_2|$,
$k_1+t_1=|L\cap K_1|$ and $k_2+t_1=|K\cap K_2|$. So, instead of
$\beta(\F,\G)$ we may write $\beta(k_1,k_2|t_1,t_2)$ to indicate the
internal angle internal angle between a  $(k_1+k_2-1)$-dimensional face $\mathcal{F}$ of $\mathcal{P}_{\w}$
with $k_1$ vertices supported on $K_1$ and $k_2$ vertices supported
on $K_2$, and  a $(k_1+k_2+t_1+t_2-1)$-dimensional face $\G$
that encompasses $\mathcal{F}$ and has $t_1+k_1$ vertices supported
on $K_1$ and the remaining $t_2+k_2$ vertices supported on $K_2$.
Similarly instead of $\zeta(\G,\mathcal{P}_{\w})$ we write
$\zeta(t_1+k_1,t_2+k_2)$ to denote the external angle between a face
$\G$ supported on set $L$ with $|L\cap K_1|=d_1$ and $|L\cap
K_2|=d_2$, and the weighted $\ell_1$-ball $\mathcal{P}_{\w}$. Using
this notation and recalling the formula (\ref{eq:angformula}) we can
write

\begin{eqnarray}
P_{K,-} &=&2\sum_{s \geq 0}\sum_{G \in \Im_{m+1+2s}(\text{SP})}
{\beta(\F,\G)\zeta(\G,\mathcal{P}_{\w})} \nonumber \\
&=& \sum_{{\tiny\begin{array}{c}0\leq t_1 \leq n_1-k_1\\0\leq t_2 \leq n_2-k_2\\ t_1+t_2 > m-k_1-k_2+1\end{array}}} 2^{t_1+t_2+1}{n_1-k_1\choose t_1}{n_2-k_2\choose t_2}\beta(k_1,k_2|t_1,t_2)\zeta(t_1+k_1,t_2+k_2),\nonumber\\
\text{} \label{eq:midsumformula}
\end{eqnarray}

\noindent where in (\ref{eq:midsumformula}) we have used the fact
that the number of faces $\G$ of $\mathcal{P}_{\w}$ of dimension
$k_1+k_2+t_1+t_2-1$ that encompass $\F$ and have $k_1+t_1$ vertices
supported on $K_1$ and its remaining $k_2+t_2$ are vertices
supported on $K_2$ is $2^{t_1+t_2}{n_1-k_1\choose
t_1}{n_2-k_2\choose t_2}$. In fact $\G$ has $k_1+k_2+t_1+t_2$
vertices including the $k_1+k_2$ vertices of $\F$. The remaining
$t_1+t_2$ vertices can each be independently in the positive or
negative orthant, therefore resulting in the term $2^{t_1+t_2}$. The
two other combinatorial terms are the number of ways one can choose
$t_1$ vertices supported on the set $K_1-K$ and $t_2$ vertices
supported on $K_2-K$.  From  (\ref{eq:midsumformula}) and (\ref{eq:error bound 2}) we can conclude theorem \ref{thm: bound on probability of error}.
\end{proof}

In the following sub-sections we will derive the internal
and external angles for a face $\F$, and a face $\G$
containing $\F$, and will provide closed form upper bounds for them.
We combine the terms together and compute the exponents using the
Laplace method in Section \ref{sec:exp cal.}, and derive thresholds
for the negativity of the cumulative exponent.

\subsubsection{Computation of Internal Angle}
\label{sec: Derivation of Inter.}
\begin{thm}\label{thm:main internal angle thm}
Let $Z$ be a random variable defined as
\begin{equation*}
Z=(k_1w_{K_1}^2+k_2w_{K_2}^2)X_1
-w_{K_1}^2\sum_{i=1}^{t_1}X_1'-w_{K_2}^2\sum_{i=1}^{t_2}X_1'',
\end{equation*}
\noindent where $X_1\sim
N(0,\frac{1}{2(k_1w_{K_1}^2+k_2w_{K_2}^2)})$ is a normal distributed
random variable, $X_i'\sim HN(0,\frac{1}{2w_{K_1}^2})~1\leq i\leq
t_1$ and $X_i''\sim HN(0,\frac{1}{2w_{K_2}^2})~1\leq i\leq t_2$ are
independent (from each other and from $X_1$) half normal distributed
random variables. Let $p_Z(\cdot)$ denote the probability
distribution function of $Z$ and $c_0=\frac{\sqrt{\pi}}{2^{l-k}}
\left((k_1+t_1)w_{K_1}^2+(k_2+t_2)w_{K_2}^2\right)^{1/2}$. Then
\begin{equation}
\label{eq:internal angle formula} \beta(k_1,k_2|t_1,t_2) = c_0p_Z(0)
\end{equation}
\end{thm}

We now prove this Theorem. Suppose that
$\F$ is a  $(k-1)$-dimensional face of the
weighted $\ell_1$-ball
\begin{equation*}
\mathcal{P}_{\w}=\{\y\in \mathbb{R}^n~|~ \sum_{i=1}^n w_i|y_i|\leq
1\}
\end{equation*}
supported on the subset $K$ with $|K|=k= k_1 + k_2$. Let $\G$ be a
$l-1$ dimensional face of $\mathcal{P}_{\w}$ supported on the set
$L$ with $\F\subset \G$. Also, let $|L\cap K_1|=k_1+t_1$ and $|L\cap
K_2|=k_2+t_2$.

We first state the following lemma the proof of which is given in
Appendix \ref{App:proof of lemma Con_F,G}.
\begin{lem}
\label{lemma:Con_F,G} Let $\F$ be a $(k-1)$-dimensional face of
$\mathcal{P}_{\w}$ supported on the set $K = \{1,2,\cdots,k\}$, and
$\G$ be a $l-1$-dimensional face of $\mathcal{P}_{\w}$ that contains
$\F$ and is supported on the set $L=\{1,2,\cdots,l\}. $Let
$\mathcal{C}_{\F^{\perp},\G}$ be the positive cone of all the
vectors $\x\in \mathbb{R}^{n}$ that take the form:
\begin{equation}
-\sum_{i=1}^{k}{b_i   e_i}+\sum_{i=k+1}^{l}{b_i   e_i},
\label{eq:vform}
\end{equation}
where $b_i, 1 \leq i \leq l$ are nonnegative real numbers  and
\begin{eqnarray*}
\sum_{i=1}^{k}{w_i b_i}=\sum_{i=k+1}^{l}{w_i b_{i}},   ~~~
\frac{b_1}{w_1}=\frac{b_2}{w_2}=\cdots=\frac{b_k}{w_k}.
\end{eqnarray*}

Then
\begin{eqnarray}
\int_{\mathcal{C}_{\F^{\perp},\G}}{e^{-\|\x\|^2}}\,d\x= \beta(\F,\G)
\cdot \pi^{(l-k)/2}. \label{eq:inaxchdirect}
\end{eqnarray}
\end{lem}
From (\ref{eq:inaxchdirect}) we can find the expression for the
internal angle. Define $U\subseteq \mathbb{R}^{l-k+1}$ as the set of
all nonnegative vectors $(x_1,x_2,\cdots,x_{l-k+1} )$ satisfying:
\begin{center}
$(\sum_{r=1}^{k}w^2_r)x_1 = \sum_{r=k+1}^{l}w^2_r x_{r-k+1}$
\end{center}
and define $f(x_1,~\cdots,~ x_{l-k+1}):U \rightarrow
\mathcal{C}_{F^{\perp},G}$ to be the following linear and bijective
map:
\begin{align*}
f(x_1,\cdots,x_{l-k+1})=-\sum_{r=1}^{k} x_1w_r \e_r+\sum_{r=k+1}^{l}
x_{r-k+1}w_r \e_r.
\end{align*}
Then

\begin{align}
\int_{ \mathcal{C}_{\F^{\perp},\G} }{e^{-\|{\x}'\|^2}}\,d{\x}' &=
\int_{U}{e^{-\|f(\x)\|^2}}\,df(\x)  = |J(\M)|
\int_{\Gamma}{e^{-\|f(x)\|^2}}\,d{x_2}
\cdots dx_{l-k+1} \nonumber \\
&=|J(\M)|
\int_{\Gamma}e^{-(\sum_{r=1}^{k}{w_r^2})x_1^2-\sum_{r=k+1}^{l}{w_r^2x_{r-k+1}^2}
} \,d{x_2} \cdots dx_{l-k+1} \label{eq:vsigular}
\end{align}

\noindent $\Gamma$ is the region described by
\begin{equation}
(\sum_{r=1}^{k}w^2_r)x_1 = \sum_{r=k+1}^{l}w^2_r x_{r-k+1} , ~x_r
\geq 0~~2 \leq r \leq l-k+1
\end{equation}
where $|J(\M)|$ is due to the change of integral variables and is
essentially the determinant of the Jacobian of the variable
transform given by the $l\times (l-k)$ matrix $\M$ below:

\begin{align}
\M_{i,j}= \left\{\begin{array}{cc} -\frac{1}{\Omega}w_iw_{k+j}^2
&{\scriptsize 1\leq i\leq k ,1\leq j\leq l-k} \\ w_{i} &
{\scriptsize k+1\leq i\leq l , j=i-k} \\ 0 & \text{Otherwise}
\end{array}\right.
\end{align}

\noindent where $\Omega = \sum_{r=1}^{k}w_r^2$. The Jacobian is
obtained by  $|J(\M)|=\det(\M^T \M)^{1/2}$. By finding the
eigenvalues of $\M^T \M$ we obtain:
\begin{equation}
|J(\M)| = w_{K_1}^{t_1}w_{K_2}^{t_2}(\frac{\Omega + t_1w_{K_1}^2 +
t_2w_{K_2}^2}{\Omega})^{1/2}
\end{equation}

\noindent Now we define a random variable
\begin{equation*}
Z = (\sum_{r=1}^{k}w^2_r)X_1 - \sum_{r=k+1}^{l}w^2_r X_{r-k+1}
\end{equation*}
where $X_1, X_2, \cdots, X_{l-k+1}$ are independent random
variables, with $X_r \sim HN(0,\frac{1}{2w_{r+k-1}^2})$, $2 \leq r
\leq (l-k+1)$, are half-normal distributed random variables and
$X_1\sim N(0, \frac{1}{2\sum_{r=1}^{k}{w_r^2}})$ is a normal
distributed random variable. Then by inspection, (\ref{eq:vsigular})
is equal to $c_1 p_{Z}(0)$, where $p_{Z}(\cdot)$ is the probability
density function for the random variable $Z$ and $p_{Z}(0)$ is the
probability density function $p_{Z}(\cdot)$ evaluated at the point
$Z=0$, and {\small
\begin{align}
c_1&=\frac{\sqrt{\pi}^{l-k+1}}{2^{l-k}}\prod_{q=k+1}^{l}\frac{1}{w_q}(\sum_{r=1}^{k}w_r^2)^{1/2}~|J(A)|
= \frac{\sqrt{\pi}^{l-k+1}}{2^{l-k}}
((k_1+t_1)w_{K_1}^2+(k_2+t_2)w_{K_2}^2)^{1/2} \label{eq:C}
\end{align}
} \normalsize Combining these results, the proof of Theorem
\ref{thm:main internal angle thm} is complete.

\subsubsection{Computation of External Angle}
\label{sec: Derivation of Extr.}

\begin{thm}\label{thm:main external angle thm}
The external angle $\zeta(\G,\mathcal{P}_{\w})=\zeta(d_1,d_2)$
between the face $\G$ and $\mathcal{P}_{\w}$ , where $G$ is
supported on the set $L$ with $|L\cap K_1|=d_1$ and $|L\cap
K_2|=d_2$ is given by:
\begin{equation}
\zeta(d_1,d_2)=\pi^{-\frac{n-l+1}{2}}2^{n-l}\int_{0}^{\infty}e^{-x^2}
\left (\int_{0}^{ \frac{w_{K_1}x}{ \xi(d_1,d_2)} } e^{-y^2}\,dy
\right)^{r_1} \left (\int_{0}^{ \frac{w_{K_2} x}{\xi(d_1,d_2)} }
e^{-y^2}\,dy \right)^{r_2}\,dx, \label{eq:external angle formula}
\end{equation}
\noindent Where $\xi^2(d_1,d_2)=\sum_{i\in
L}w_i^2=d_1w_{K_1}^2+d_2w_{K_2}^2$, $r_1=n_1-d_1$ and $r_2=n_2-d_2$.
\end{thm}
\begin{proof}
Without loss of generality, assume that the support set of $\G$ is
given by $L=\{n-l+1,n-l+2,\cdots,n\}$ and consider the
$(l-1)$-dimensional face
\begin{equation*}
\G=\text{conv}\{\frac{\e_{n-l+1}}{w_{n-l+1}} , ...
,\frac{\e_{n-k}}{w_{n-k}} , \frac{\e_{n-k+1}}{w_{n-k+1}},
...,\frac{\e_{n}}{w_{n}} \}
\end{equation*}
of the weighted $\ell_1$-ball $\mathcal{P}$. The $2^{n-l}$ outward
normal vectors of the supporting hyperplanes of the facets
containing $\G$ are given by
\begin{equation*}
\{\sum_{i=1}^{n-l} j_{i}w_i \e_i+\sum_{p=n-l+1}^{n} w_i \e_i,
j_{i}\in\{-1,1\}\}.
\end{equation*}

\noindent Then the outward normal cone $\mathcal{C}_{\G,
\mathcal{P}_{\w}}^\perp$ at the face $\G$ is the positive hull of
these normal vectors. Thus
\begin{align}
\int_{\mathcal{C}_{\G,
\mathcal{P}_{\w}}^\perp}{e^{-\|x\|^2}}\,dx&=\zeta(\G,\mathcal{P}_{\w})
V_{n-l}(S^{n-l}) \int_{0}^{\infty}{e^{-r^2}}r^{n-l}\,dx \nonumber\\
&=\zeta(\G,\mathcal{P}_{\w}).\pi^{(n-l+1)/2},
 \label{eq:axch}
\end{align}
where $V_{n-l}(S^{n-l})$ is the spherical volume of the
$(n-l)$-dimensional unit sphere $S^{n-l}$. 
Now define $U$ to be the set
\begin{equation*}
\{x \in R^{n-l+1} \mid x_{n-l+1} \geq 0, | x_i/w_i| \leq x_{n-l+1},
1\leq i \leq (n-l)\}
\end{equation*}
and define $f(x_1,~\cdots,~ x_{n-l+1}):U \rightarrow
\mathcal{C}_{\G, \mathcal{P}_{\w}}^\perp$ to be the linear and
bijective map
\begin{eqnarray*}
f(x_1,~\cdots,~ x_{n-l+1})&=&\sum_{i=1}^{n-l} x_i
\e_i+\sum_{i=n-l+1}^{n} w_i x_{n-l+1}\e_i .
\end{eqnarray*}
Then \small{
\begin{align}
&\int_{\mathcal{C}_{\G, \mathcal{P}_{\w}}^\perp}{e^{-\|x'\|^2}}\,dx' = |J(\M)|\int_{U}{e^{-\|f(x)\|^2}}\,dx \nonumber\\
&=|J(\M)|\int_{0}^{\infty}
\int_{-w_{K_1}x_{n-l+1}}^{w_{K_1}x_{n-l+1}}
\cdots\int_{-w_{n-l}x_{n-l+1}}^{w_{n-l}x_{n-l+1}} e^{-x_1^2-\cdots
-x_{n-l}^2-(\sum_{i=n-l+1}^{n}w_i^2)x_{n-l+1}^{2} } \,dx_{1} \cdots \,dx_{n-l+1}\nonumber\\
&=|J(\M)|\int_{0}^{\infty} e^{-(\sum_{i=n-l+1}^{n}w_i^2)x^2}
\left(\int_{-w_{K_1} x}^{w_{K_1} x} e^{-y^2}\,dy \right )^{n_1-d_1}
\left(\int_{-w_{K_2}x}^{w_{K_2}x} e^{-y^2}\,dy \right )^{n_2-d_2}
\,dx \label{eq: aux external angle formula}
\end{align}
} \normalsize $\M$ is the $n\times (n-l+1)$ change of variable
matrix given by $ \M=\left(\begin{array}{cc} \bf I_{n-l} & 0 \\ 0 &
\w_{L}
\end{array}\right)$,
\noindent where $\w_{L}=(w_{n-l+1},w_{n-l+2},\cdots,w_n)^T$.
Therefore
${J(\M)=\det(\M^T\M)}^{1/2}=(d_1w_{K_1}^2+d_2w_{K_2}^2)^{1/2}$.
Replacing this and a change of variable for $x$ (replace $\xi x$
with $x$) in (\ref{eq: aux external angle formula}), along with
(\ref{eq:axch}), complete the proof.
\end{proof}

\subsubsection{Derivation of the Critical Weak and Strong $\delta_c$ Threshold }
\label{sec:exp cal.}

So far we  have proved that the probability of the failure event is bounded
by the formula \small{
\begin{align}
&\Prob\{E^c\} \leq \sum_{{\tiny\begin{array}{c}0\leq t_1 \leq n_1-k_1\\0\leq t_2 \leq
n_2-k_2\\ t_1+t_2 > m-k_1-k_2+1\end{array}}} 2^{t_1+t_2+1}{n_1-
k_1\choose t_1}{n_2-k_2\choose t_2}
\beta(k_1,k_2|t_1,t_2)\zeta(t_1+k_1,t_2+k_2),
\label{eq:sumformula2}
\end{align}
} 
\normalsize where we gave expressions for  $\beta(t_1,t_2|k_1,k_2)$ and $\zeta(t_1+k_1,t_2,k_2)$  in Sections \ref{sec: Derivation of
Inter.} and \ref{sec: Derivation of Extr.}, respectively. Now our objective is to show that the R.H.S
of (\ref{eq:sumformula2}) will exponentially decay to $0$ as
$n\rightarrow\infty$, provided that $\delta=\frac{m}{n}$ is greater
than a critical threshold $\delta_c$, which we are trying to
evaluate. To do this end we bound the exponents of the
combinatorial, internal angle and external angle terms in
(\ref{eq:sumformula2}), and find the values of $\delta$ for which
the net exponent is strictly negative. The maximum such $\delta$
will give us $\delta_c$. Starting with the combinatorial term, we
use Stirling approximating on the binomial coefficients to achieve
the following as $n\rightarrow \infty$ and $\epsilon\rightarrow 0$
\small{
\begin{equation}
\frac{1}{n}\log\left(2^{t_1+t_2+1}{n_1- k_1\choose
t_1}{n_2-k_2\choose t_2}\right)\rightarrow
\left(\gamma_1(1-p_1)H(\frac{\tau_1}{\gamma_1(1-p_1)})+\gamma_2(1-p_2)H(\frac{\tau_2}{\gamma_2(1-p_2)})+\tau_1+\tau_2\right)\log{2},
\label{eq:comb exp.}
\end{equation}
} \normalsize

\noindent where $\tau_1=\frac{t_1}{n}$ and $\tau_2=\frac{t_2}{n}$.

\noindent For the external angle and internal angle terms we prove
the following two  exponents

\begin{enumerate}
\item Let  $g(x)=\frac{2}{\sqrt{\pi}}e^{-{x^2}}$, $G(x)=\frac{2}{\sqrt{\pi}}\int_{0}^{x}e^{-y^2}dy$. Also define $c=(\tau_1+\gamma_1p_1)+\omega^2(\tau_2+\gamma_2p_2)$, $\alpha_1=\gamma_1(1-p_1)-\tau_1$ and $\alpha_2=\gamma_2(1-p_2)-\tau_2$. Let $x_0$ be the unique solution to $x$ of the following:
\begin{equation*}
2c-\frac{g(x)\alpha_1}{xG(x)}-\frac{\omega g(\omega
x)\alpha_2}{xG(\omega x)}=0
\end{equation*}
Define
\begin{equation}
\psi_{ext}(\tau_1,\tau_2) =
cx_0^2-\alpha_1\log{G(x_0)}-\alpha_2\log{G(\omega x_0)}
\label{eq:ext exp.}
\end{equation}

\item Let $b=\frac{\tau_1+\omega ^2\tau_2}{\tau_1+\tau_2}$ and  $\varphi(.)$ and $\Phi(.)$ be the standard Gaussian pdf and cdf functions respectively. Also let $\Omega'=\gamma_1p_1+\omega ^2\gamma_2p_2$ and $Q(s)=\frac{\tau_1\varphi(s)}{(\tau_1+\tau_2)\Phi(s)}+\frac{\omega \tau_2\varphi(\omega s)}{(\tau_1+\tau_2)\Phi(\omega s)}$. Define the function $\hat{M}(s)=-\frac{s}{Q(s)}$ and solve for $s$ in $\hat{M}(s)=\frac{\tau_1+\tau_2}{(\tau_1+\tau_2)b+\Omega'}$. Let the unique solution be $s^*$ and set $y=s^*(b-\frac{1}{\hat{M}(s^*)})$. Compute the rate function $\Lambda^*(y)= sy -\frac{\tau_1}{\tau_1+\tau_2}\Lambda_1(s)-\frac{\tau_2}{\tau_1+\tau_2}\Lambda_1(\omega s)$ at the point $s=s^*$, where $\Lambda_1(s) = \frac{s^2}{2} +\log(2\Phi(s))$.
The internal angle exponent is then given by:
\begin{equation}
\psi_{int}(\tau_1,\tau_2) =
(\Lambda^*(y)+\frac{\tau_1+\tau_2}{2\Omega'}y^2+\log2)(\tau_1+\tau_2).
\label{eq:int exp.}
\end{equation}

\end{enumerate}
\noindent We now state the following lemmas, which are proved in  Appendix \ref{App: proof of lemma
ext. exp.} and \ref{App: proof of lemma ext. int.}.

\begin{lem}
Fix $\delta$, $\epsilon > 0$. There exists a finite number
$n_0(\delta,\epsilon)$ such that
  \begin{equation}
  \frac{1}{n} \log(\zeta(t_1+k_1,t_2+k_2)) <- \psi_{ext}(\tau_1,
\tau_2)+\epsilon,
  \end{equation}
uniformly in $0\leq t_1 \leq n_1-k_1$, $0\leq t_2 \leq n_2-k_2$ and
$t_1+t_2 \geq m-k_1-k_2+1$, $n \geq n_0(\delta, \epsilon)$.
\label{lemma:extasy}
\end{lem}

\begin{lem}
Fix $\delta$, $\epsilon > 0$. There exists a finite number
$n_1(\delta,\epsilon)$ such that
  \begin{equation}
  \frac{1}{n} \log(\beta(t_1,t_2|k_1,k_2)) <- \psi_{int}(\tau_1,
\tau_2)+\epsilon,
  \end{equation}
uniformly in $0\leq t_1 \leq n_1-k_1$, $0\leq t_2 \leq n_2-k_2$ and
$t_1+t_2 \geq m-k_1-k_2+1$, $n \geq n_1(\delta, \epsilon)$.
\label{lemma:intasy}
\end{lem}
 
Combining Lemmas \ref{lemma:extasy} and \ref{lemma:intasy}, (\ref{eq:comb exp.}), and
the bound in (\ref{eq:sumformula2}) we readily get the critical
bound for $\delta_c$ as in the Theorem \ref{thm:main delta bound}.

Derivation of the strong and sectional threshold can be easily done using union bounds to account for all possible support sets and/or all sign patterns. The corresponding upper bound on the failure probability for the strong threshold is given by:
\beq {n_1\choose k_1}{n_2\choose k_2}2^{k}P_{K,-}\eeq
\noindent It then follows that the strong threshold of $\delta$ is given by $\delta_c$ in Theorem \ref{thm:main delta bound}, except that the combinatorial exponent $\psi_{com}(\cdot,\cdot)$ must be corrected by adding a term
\beq (\gamma_1p_1 + \gamma_2p_2 + \gamma_1H(p_1) +\gamma_2H(p_2))\log2, \eeq

\noindent to the RHS of (\ref{eq:comb angle}). Similarly, for the sectional threshold, which deals with all possible support sets but almost all sign patterns, the modification in the combinatorial exponent term is as follows:

\beq (\gamma_1H(p_1) +\gamma_2H(p_2))\log2. \eeq

\subsection{Generalizations}
\label{sec:generalizations}
Except for some subtlety in the large deviation calculations, the generalization of the results of the previous section to an arbitrary $u\geq 2$ classes of entries is straightforward. Consider a nonuniform sparse model with $u$ classes $K_1,\cdots,K_u$ where $|K_i|=n_i = \gamma_in$, and the sparsity fraction over the set $K_i$ is $p_i$, and a recovery scheme based on weighted $\ell_1$ minimization with weight $\omega_i$ for the set $K_i$. The bound in (\ref{eq:angformula}) is general and can always be used. Due to isomorphism, the internal and external angles $\beta(\F,\G)$ and $\zeta(\G,\mathcal{P}_{\w})$ only depend on the number of vertices that the supports of $\F$ and $\G$ have in common with each $K_i$. Therefore, a generalization to (\ref{eq:sumformula}) would be:
\small{
\begin{align}
\Prob\{E^c\} &\leq 2\sum_{{\tiny\begin{array}{c}0 \leq \t  \leq \n-\k \\ \mathbf{1}^T \t > m-\mathbf{1}^T \k+1\end{array}}} \Pi_{1\leq i\leq u} 2^{t_i}{n_i-
k_i\choose t_i}
\beta(\k|\t)\zeta(\t+\k)
\label{eq:sumformula_generalized}
\end{align}
}
\normalsize
\noindent Where $\t = (t_1,\cdots,t_u)^T$,  $\k = (k_1,\cdots,k_u)^T$ and $\mathbf{1}$ is a vector of all ones. Invoking generalized forms of Theorems \ref{thm:main external angle thm} and \ref{thm:main internal angle thm} to approximate the terms $\beta(\k|\t)$ and $\zeta(\k+\t)$, we conclude the following Theorem.

\begin{thm}\label{thm:main delta bound_generalized}
Consider a nonuniform sparse model with $u$ classes $K_1,\cdots,K_u$ with $|K_i| = n_i = \gamma_1n$, and sparsity fractions $p_1,p_2,\cdots,p_u$, where $n$ is the signal dimension. Also, let the functions $g(.),G(.),\psi(.),\Psi(.)$ be as defined in Theorem \ref{thm:main delta bound}. For positive values $\{\omega_i\}_{i=1}^u$, the recovery thresholds (weak,sectional and strong) of the weighted $\ell_1$ minimization program:
\beq \min_{\A\x=\y}\sum_{i=1}^u \omega_i\|\x_{K_i}\|_1,\nonumber\eeq
\noindent is given by the following expression:
\beq \addtolength{\fboxsep}{5pt} \boxed{
\begin{gathered}
\delta_c = \min\{\delta~|~\psi_{com}(\mathbf{\tau})-\psi_{int}(\mathbf{\tau})-\psi_{ext}(\mathbf{\tau})<0~ \forall \mathbf{\tau}=(\tau_1,\cdots,\tau_u)^T:\nonumber \\ ~0\leq \tau_i\leq \gamma_i(1-p_i)\forall 1\leq i\leq u,
\sum_{i=1}^u\tau_i  > \delta-\sum_{i=1}^u\gamma_ip_i \} \nonumber
\end{gathered}
} \eeq \noindent where $\psi_{com}$, $\psi_{int}$ and $\psi_{ext}$
are obtained from the following expressions:
\begin{enumerate}
\item $\psi_{com}(\mathbf{\tau}) =
\log{2}\sum_{i=1}^u\gamma_i(1-p_i)H(\frac{\tau_i}{\gamma_i(1-p_i)})+\tau_i,
$ for the weak threshold. For sectional threshold this must be modified by adding a term $\log2\sum_{i=1}^u\gamma_iH(p_i)$. For strong threshold, it must be also added with $\sum_{i=1}^u \gamma_ip_i$.
\item $\psi_{ext}(\mathbf{\tau}) =cx_0^2-\sum_{i=1}^u\alpha_i\log{G(\omega_ix_0)} $, where $c=\sum_{i=1}^u \omega_i^2(\tau_i+\gamma_ip_i)$,
$\alpha_i=\gamma_i(1-p_i)-\tau_i$ and $x_0$ is the unique solution
  of $ 2c=\sum_{i=1}^u\omega_i\frac{g(\omega_ix_0)\alpha_i}{x_0G(\omega x_0)},$

\item $\psi_{int}(\mathbf{\tau})=\lambda(\Lambda^*(y)+\frac{\lambda y^2}{2\sum_{i=1}^u\omega_i^2\gamma_ip_i}+\log2)$, where $\lambda = \sum_{i=1}^u\tau_i$, and $y$ and $\Lambda^*(y)$ are obtained as follows.
Let  $b=\frac{\sum_{i=1}^u\omega_i^2\tau_i}{\lambda}$,
$Q(s)=\sum_{i=1}^u\frac{\tau_i\varphi(s)}{\lambda\Phi(s)}$. Let $s^*$ be the solution to $s$ in
$-\frac{Q(s)}{s}=b + \frac{\sum_{i=1}^u\omega_i^2\gamma_ip_i}{\lambda}$, and $y=s^*(b-\frac{1}{\hat{M}(s^*)})$.
Then $\Lambda^*(y)= s^*y
-1/\lambda\sum_{i=1}^u {\tau_i \left(\frac{\omega_i^2 {s^*}^2}{2} +\log(2\Phi(\omega_is^*)) \right)}$ .
\end{enumerate}
\end{thm}

\subsection{Robustness}
\label{sec:robustness}

\begin{proof}[proof of Theorem \ref{thm:robustness}.]
We first state the following lemma, which is very similar to Theorem 2 of \cite{Weiyu GM}. We skip its proof for brevity.
\begin{lem}
Let $K\subset\{1,2\cdots,n\}$ and the weight vector $\w = (w_1,w_2,\cdots,w_n)^T$ be fixed. Define $\W = diag(w_1,w_2,\cdots,w_n)$  and suppose $C>1$ is given. For every vector $\x_0\in \mathbb{R}^{n\times1}$, the solution $\hat{\x}$  of (\ref{eq:weighted l_1}) satisfies
\beq \|\W(\x_0-\hat{\x})\|_1\leq 2\frac{C+1}{C-1}\sum_{i\in\overline{K}}w_i|(x_0)_i|,\eeq
\noindent if and only if for every $\z\in\mathcal{N}(\A)$ the following holds:
\beq C \sum_{i\in K}w_i|z_i| \leq \sum_{i\in \overline{K}}w_i|z_i|.\eeq
\label{lem:robustness}
\end{lem}

\noindent Let $\z = (z_1,\cdots,z_n)^T$ be a vector in the null space of $\A$, and assume that
\beq C'\sum_{i\in L_1 \cup L_2}w_i|z_i| = \sum_{i\in \overline{L_1} \cap \overline{L_2}}w_i|z_i|. \label{eq:C'}\eeq

\noindent Let $K_{\epsilon_1}$ and $K_{\epsilon_2}$ be the solutions of the following  problems
\bea K_{\epsilon_1}&:&  \max_{K_{\epsilon_1}\subset K_1\cap\overline{L_1}, |K_{\epsilon_1}| = \epsilon_1\gamma_1p_1n} \sum_{i\in K_{\epsilon_1}}w_i|z_i|,\\
K_{\epsilon_2}&:&  \max_{K_{\epsilon_2}\subset K_2\cap\overline{L_2}, |K_{\epsilon_2}| = \epsilon_2\gamma_2p_2n} \sum_{i\in K_{\epsilon_2}}w_i|z_i|. \eea

\noindent Let $L'_1 = L_1\cup K_{\epsilon_1}$ and $L'_2 = L_2\cup K_{\epsilon_2}$. From the definition of $K_{\epsilon_1}$ and $K_{\epsilon_2}$, it follows that
\bea \sum_{i\in K_{\epsilon_1}} w_i|z_i| &\geq& \frac{\epsilon_1p_1}{1-p_1} \sum_{i\in \overline{L'_1}\cap K_1} w_i|z_i|,  \label{eq:sum_K_eps1}\\
\sum_{K_{i\in \epsilon_2}} w_i|z_i| &\geq& \frac{\epsilon_2p_2}{1-p_2} \sum_{i\in \overline{L'_2}\cap K_2} w_i|z_i|.  \label{eq:sum_K_eps2}
\eea
\noindent Adding $C'\left(\sum_{K_{\epsilon_1}} w_i|z_i| + \sum_{K_{\epsilon_2}} w_i|z_i| \right)$ to both sides of (\ref{eq:C'}) and using (\ref{eq:sum_K_eps1}) and (\ref{eq:sum_K_eps2}), we can write:

\bea C'\sum_{i\in L'_1 \cup L'_2}w_i|z_i| &\geq&  \sum_{i\in \overline{L_1} \cap \overline{L_2}}w_i|z_i| + C'\left(\frac{\epsilon_1p_1}{1-p_1} \sum_{i\in \overline{L'_1}\cap K_1} w_i|z_i| + \frac{\epsilon_2p_2}{1-p_2} \sum_{i\in \overline{L'_2}\cap K_2} w_i|z_i|  \right) \\
&\geq& \left(1 + (C'+1)\min(\frac{\epsilon_1p_1}{1-p_1},\frac{\epsilon_2p_2}{1-p_2})\right)   \sum_{i\in \overline{L'_1} \cap \overline{L'_2}}w_i|z_i|. \label{eq:contradict}
\eea
\noindent Note that $|L'_1| = \gamma_1p_1n$ and $|L'_2| = \gamma_2p_2n$. Therefore, since $\delta =\frac{m}{n}\geq \delta_c^{(S)}(\gamma_1,\gamma_2,p_1,p_2,\omega)$, we know that $\sum_{i\in L'_1 \cup L'_2}w_i|z_i| \leq  \sum_{i\in \overline{L'_1} \cap \overline{L'_2}}w_i|z_i| $. From this and (\ref{eq:contradict}) we conclude that
\beq C' \geq  \left(1 + (C'+1)\min(\frac{\epsilon_1p_1}{1-p_1},\frac{\epsilon_2p_2}{1-p_2})\right),  \eeq
\noindent or equivalently
\beq C' \geq \frac{1 + \min(\frac{\epsilon_1p_1}{1-p_1},\frac{\epsilon_2p_2}{1-p_2})}{1 -   \min(\frac{\epsilon_1p_1}{1-p_1},\frac{\epsilon_2p_2}{1-p_2})}.\eeq
Using Lemma \ref{lem:robustness} and the above inequality, we conclude (\ref{eq:robustness}).
\noindent
\end{proof}

\section{Approximate Support Recovery and Reweighted $\ell_1$}
\label{sec:approximated supp recov}
Using the analytical tools of this paper, it is possible to prove that a class of reweighted  $\ell_1$ minimization algorithms have a strictly higher recovery thresholds for sparse signals whose nonzero entries follow certain classes
of distributions (e.g. Gaussian). The technical details of this claim is not brought here, since it stands beyond the scope of this paper.  However, we briefly mention how a simple post processing on the output of $\ell_1$ minimization results in a nonuniform sparsity model with $u=2$ classes close to the one we introduced for the unknown signal. A more comprehensive study on this can be found in \cite{ISIT10 Reweighted}. The reweighted $\ell_1$ recovery algorithm proposed in \cite{ISIT10 Reweighted} is composed of two steps. In the first step a standard $\ell_1$ minimization is done, and based on the output, a set of entries where the signal is likely to reside (the so-called approximate support) is identified. The unknown signal can thus be thought of as  two classes, one with a relatively high fraction of nonzero entries, and one with a small fraction. The second step is a weighted $\ell_1$ minimization step where entries outside the approximate support set are penalized with a constant weight larger than $1$. The algorithm is as follows:

\begin{alg}
\text{}
\begin{enumerate}
\item Solve the $\ell_1$ minimization problem:
\begin{equation}
\hat{\x} = \arg{ \min{ \|\z\|_1}}~~\text{subject to}~~ \A\z
= \A\x.
\end{equation}
\item Obtain an approximation for the support set of $\x$:
find the index set $L \subset \{1,2, ..., n\}$ which corresponds to
the largest $k$ elements of
$\hat{\x}$ in magnitude.
\item Solve the following weighted $\ell_1$ minimization problem and declare the solution as output:
 \beq
\x^* = \arg{\min\|\z_L\|_1+\omega\|\z_{\overline{L}}\|_1}~~\text{subject to}~~ \A\z
= \A\x.
\label{eq:weighted l_1Alg}
\eeq
\end{enumerate}
\label{alg:modmain}
\end{alg}

For a given number of measurements, if the support size of $\x$, namely $k=|K|$, is slightly larger than the sparsity threshold of $\ell_1$ minimization, then a so-called robustness of $\ell_1$ minimization helps find a lower bound $f_1$ for $\frac{|L\cap K|}{|L|}$, i.e. the sparsity fraction of $\x$ over the set $L$. If $f_1$ is sufficiently close to 1, the number of measurements could satisfy:
\beq \delta \geq \max_{f'_1\geq f_1, f'_1k+f'_2(n-k) = k}\delta_c^{(T)}(\frac{k}{n},1-\frac{k}{n},f'_1,f'_2,\omega).\eeq
\noindent Then the recovery is successful in the second step with high probability. Recall that $\delta_c^{(T)}$ is the sectional threshold, which accounts for  all possible support sets. Therefore, the condition for strict improvement in the reweighted $\ell_1$ minimization is that:
\beq \delta(0,1,0,\frac{k}{n},1)\geq \max_{f'_1\geq f_1, f'_1k+f'_2(n-k) = k}\delta_c^{(T)}(\frac{k}{n},1-\frac{k}{n},f'_1,f'_2,\omega).\eeq

\section{Simulation Results}
\label{sec:Simulation}

\begin{figure*}[t]
  \centering
  \subfloat[$\gamma_1 = \gamma_2 = 0.5$, $p_1=0.4$ and $p_2=0.05$.]{\label{fig:simulation threshold_1}\includegraphics[width=0.5\textwidth]{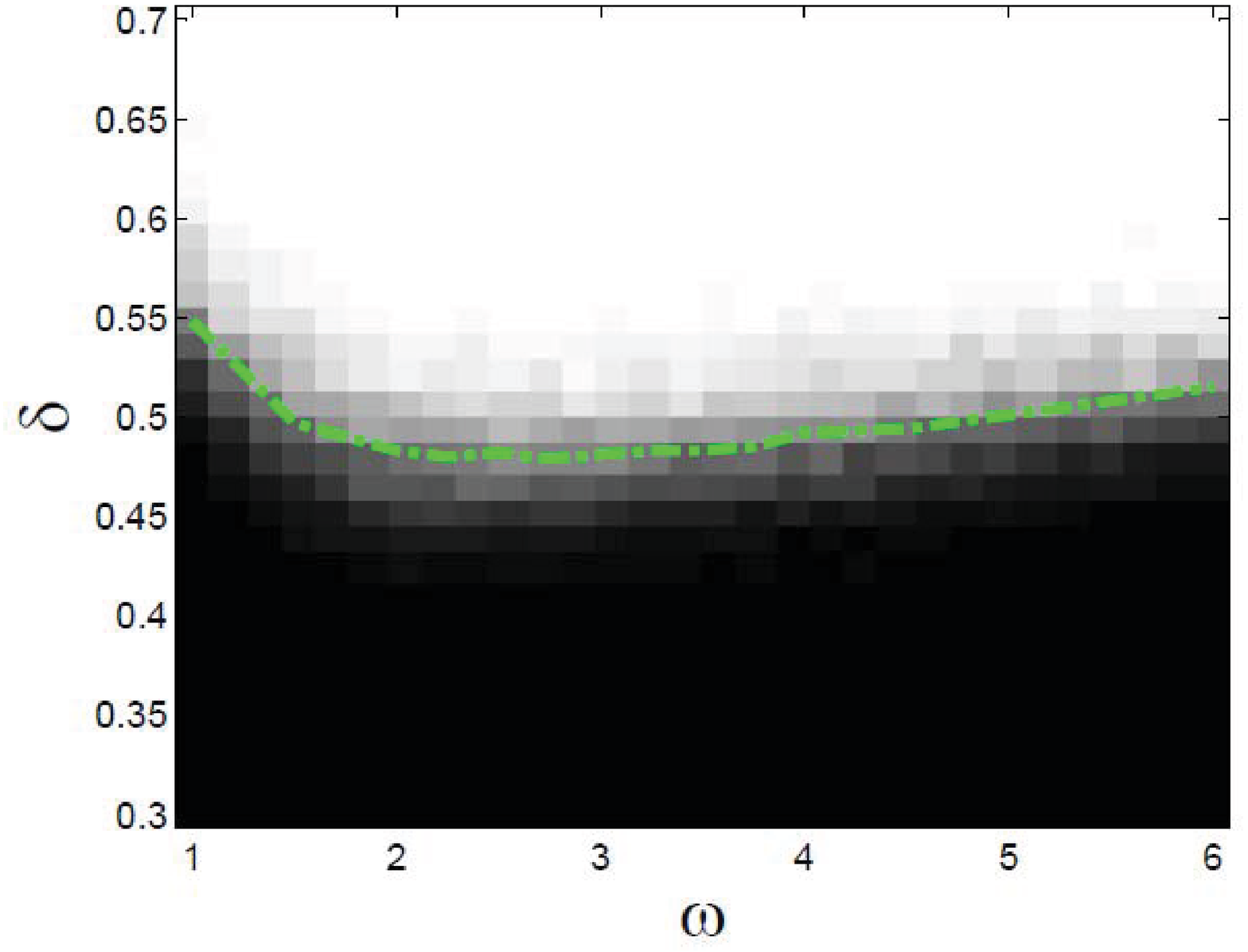}}
  \subfloat[$\gamma_1 = \gamma_2 = 0.5$, $p_1=0.65$ and $p_2=0.1$.]{\label{fig:simulation threshold_2}\includegraphics[width=0.5\textwidth]{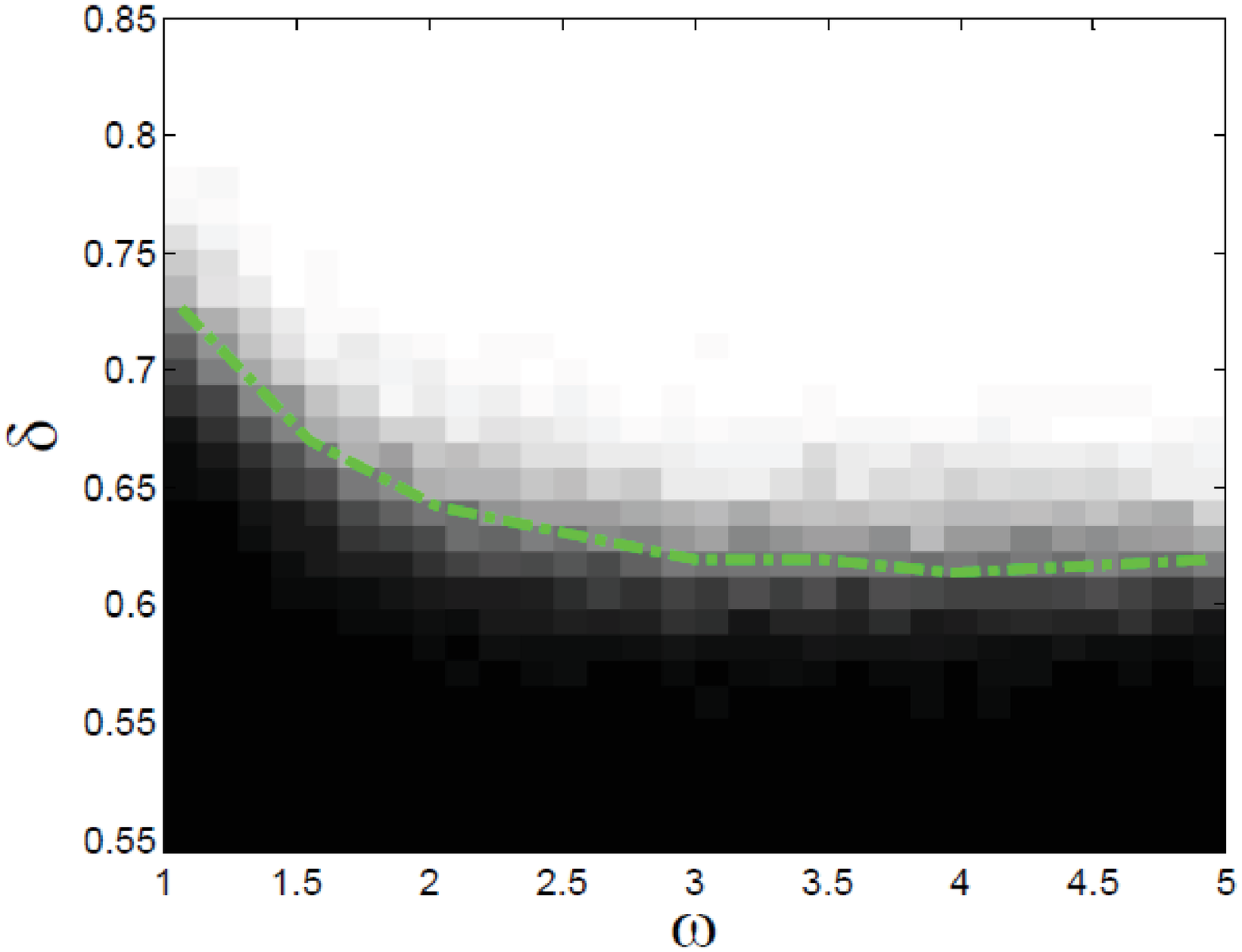}}
    \caption{{\scriptsize  Empirical recovery percentage of weighed $\ell_1$ minimization for different weight values $\omega$, and different number of measurements $\delta = \frac{m}{n}$ and $n=200$. Signals have been selected from a nonuniform sparse models. White indicates perfect recovery..}}
  \label{fig:simulation threshold}
\end{figure*}

%

We demonstrate by some examples that appropriate weights can boost
the recovery percentage. In Figure \ref{fig:simulation threshold} we
have shown the empirical recovery threshold of weighted $\ell_1$
minimization for different values of the weight
$\omega=\frac{w_{K_1}}{w_{K_2}}$, for two particular nonuniform
sparse models. Note that the empirical threshold is somewhat
identifiable with naked eye, and is very similar to the theoretical
curve of  Figure \ref{fig:delta_c} for similar settings. In another
experiment, we fix $p_2$ and  $n=2m=200$, and try $\ell_1$ and
weighted $\ell_1$ minimization for various values of $p_1$. We
choose $n_1=n_2=\frac{n}{2}$. Figure \ref{fig:PRvsP1firstsetting}
shows one such comparison for $p_2=0.05$ and different values of
$w_{K_2}$. Note that the optimal value of $w_{K_2}$ varies as $p_1$
changes. Figure \ref{fig:PRvsP1Optimalw} illustrates how the optimal
weighted $\ell_1$ minimization surpasses the ordinary $\ell_1$
minimization. The optimal curve is basically achieved by selecting
the best weight of  Figure \ref{fig:PRvsP1firstsetting} for each
single value of $p_1$. Figure \ref{fig:PRvsP1} shows the result of
simulations in another setting where $p_2=0.1$ and $m=0.75n$
(similar to the setting of Section \ref{sec: Main results}).  Note that these
results very well match the theoretical results of Figures
\ref{fig:critical delta} and \ref{fig:critical delta other setting}.

\begin{figure*}[t]
  \centering
  \subfloat[]{\label{fig:PRvsP1firstsetting}\includegraphics[width=0.45\textwidth]{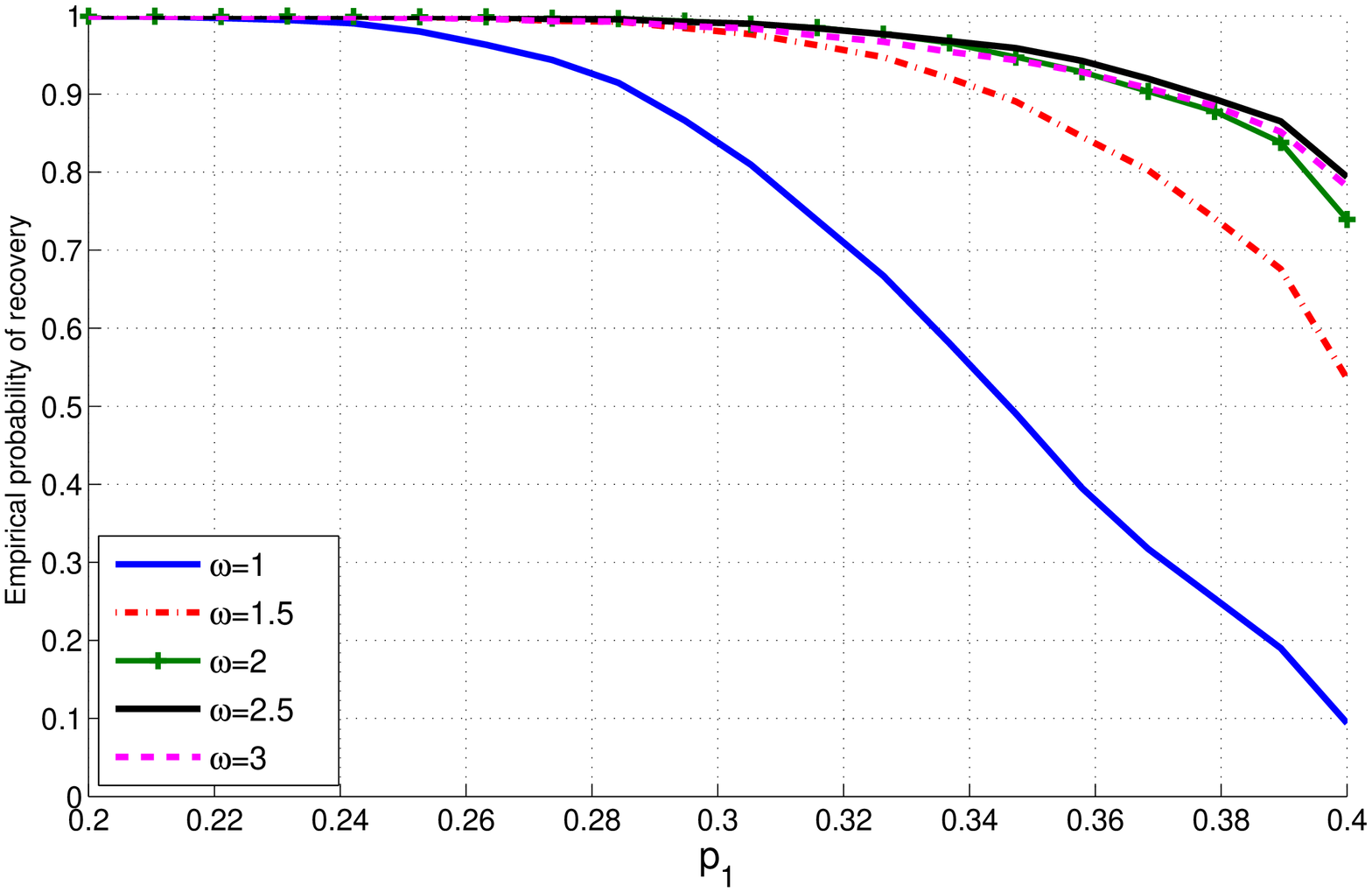}}
  \subfloat[]{\label{fig:PRvsP1Optimalw}\includegraphics[width=0.45\textwidth]{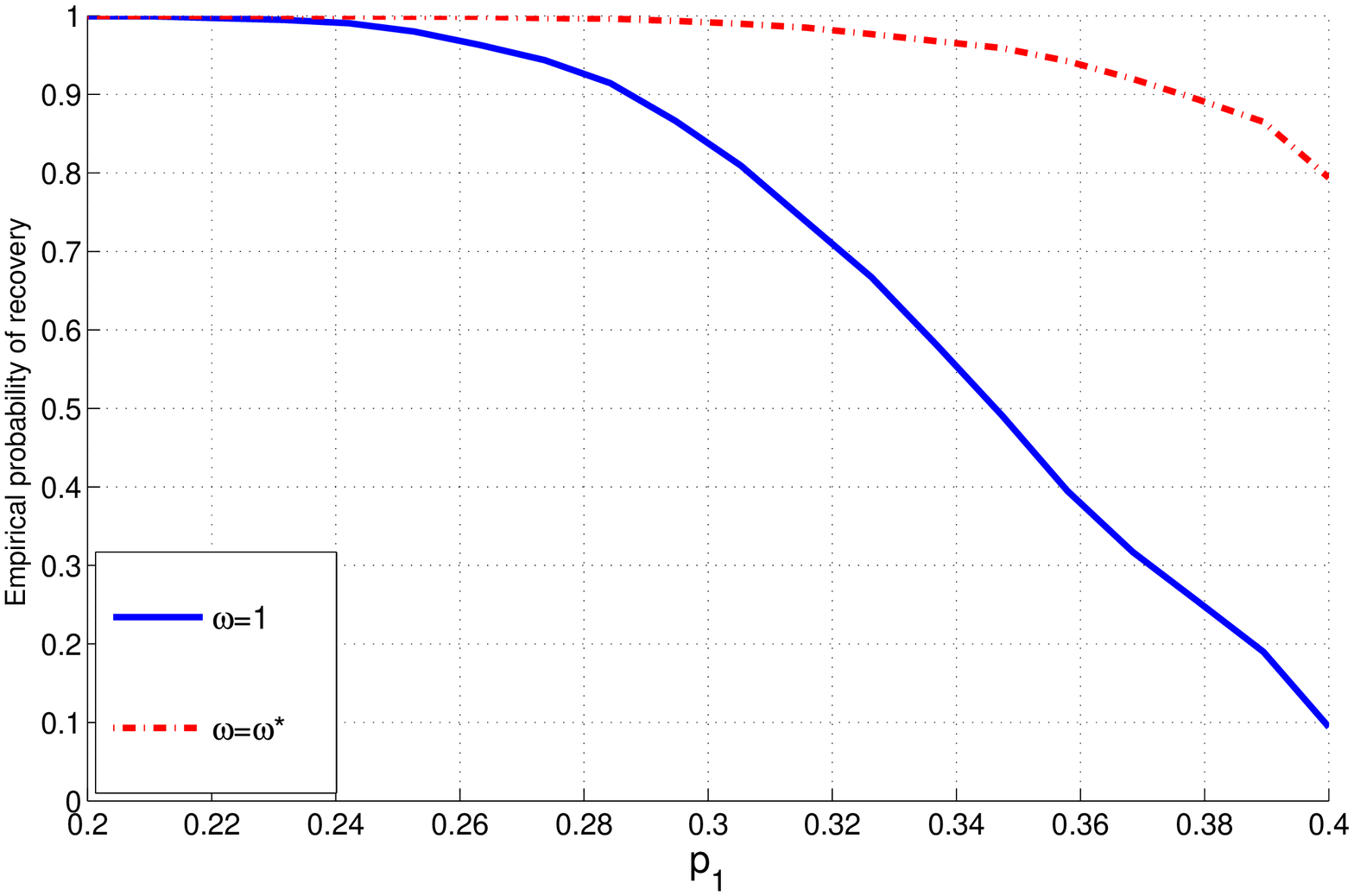}}
    \caption{{\scriptsize Empirical probability of successful recovery for weighted $\ell_1$ minimization with different weights (unitary weight for the first subclass and $\omega$ for the other one) and suboptimal weights in a nonuniform sparse setting. $p_2=0.05$, $\gamma_1=\gamma_2=0.5$ and $m=0.5n = 100$.  $\omega^*$ is (b) is the optimum value of $\omega$ for each $p_1$  among the values shown in (a).}}
  \label{fig:weightedimprovement}
\end{figure*}

\begin{figure}[t]
\centering
  \includegraphics[width=0.6\textwidth]{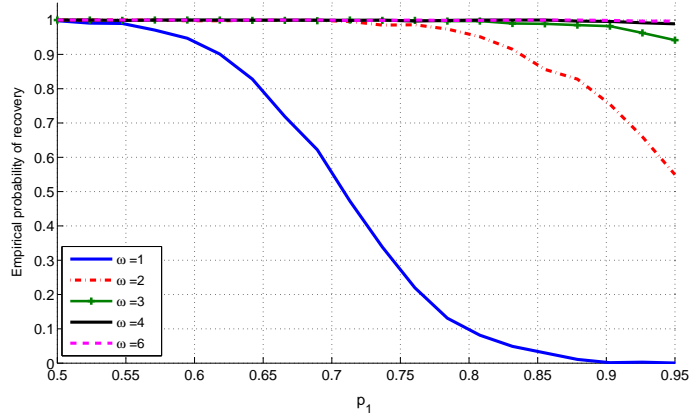}\\
\caption{{\scriptsize Empirical probability of successful recovery  for different
weights. $p_2=0.1$,$\gamma_1=\gamma_2=0.5$ and $m=0.75n=150$.}} \label{fig:PRvsP1}
\end{figure}

In Figure \ref{fig:avplots}, we have displayed the performance of weighted $\ell_1$ minimization in the presence of noise. The original signal is a nonuniformly sparse vector with sparsity fractions $p_1 = 0.4, p_2=0.05$ over two subclasses $\gamma_1=\gamma_2=0.5$. However, a white Gaussian noise vector is added before compression. Figure \ref{fig:avplots} shows a scatter plot of all output signal to recovery error ratios as a function of the input SNR, for all simulations. In Figure \ref{fig:allavplots} the average curves are compared together for different values of weight $\omega$.

\begin{figure}
  \centering
  \subfloat[$\omega=1$.]{\label{fig:avplotW=1}\includegraphics[width=0.4\textwidth]{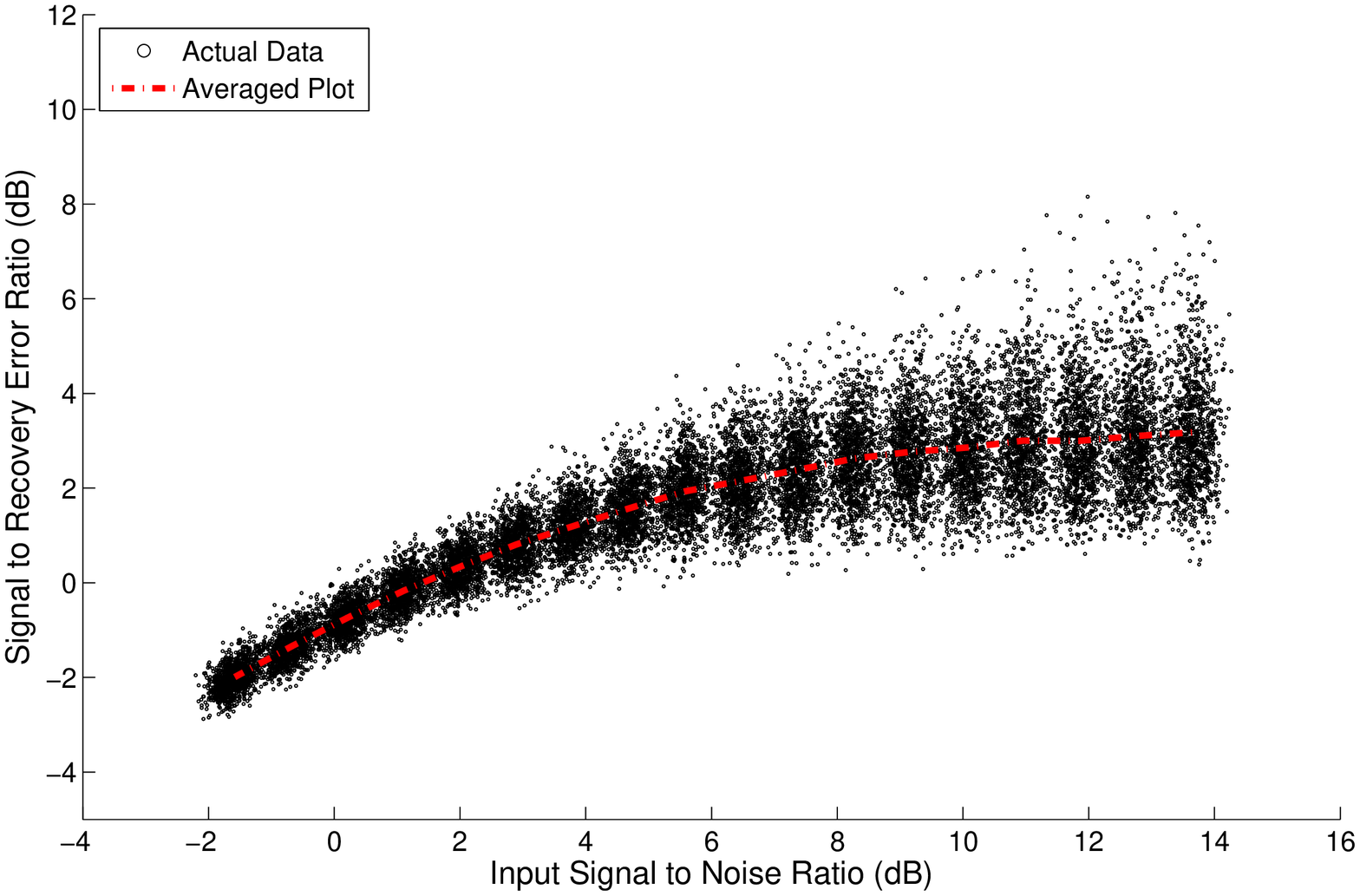}}
  \subfloat[$\omega=3$.]{\label{fig:avplotW=3}\includegraphics[width=0.4\textwidth]{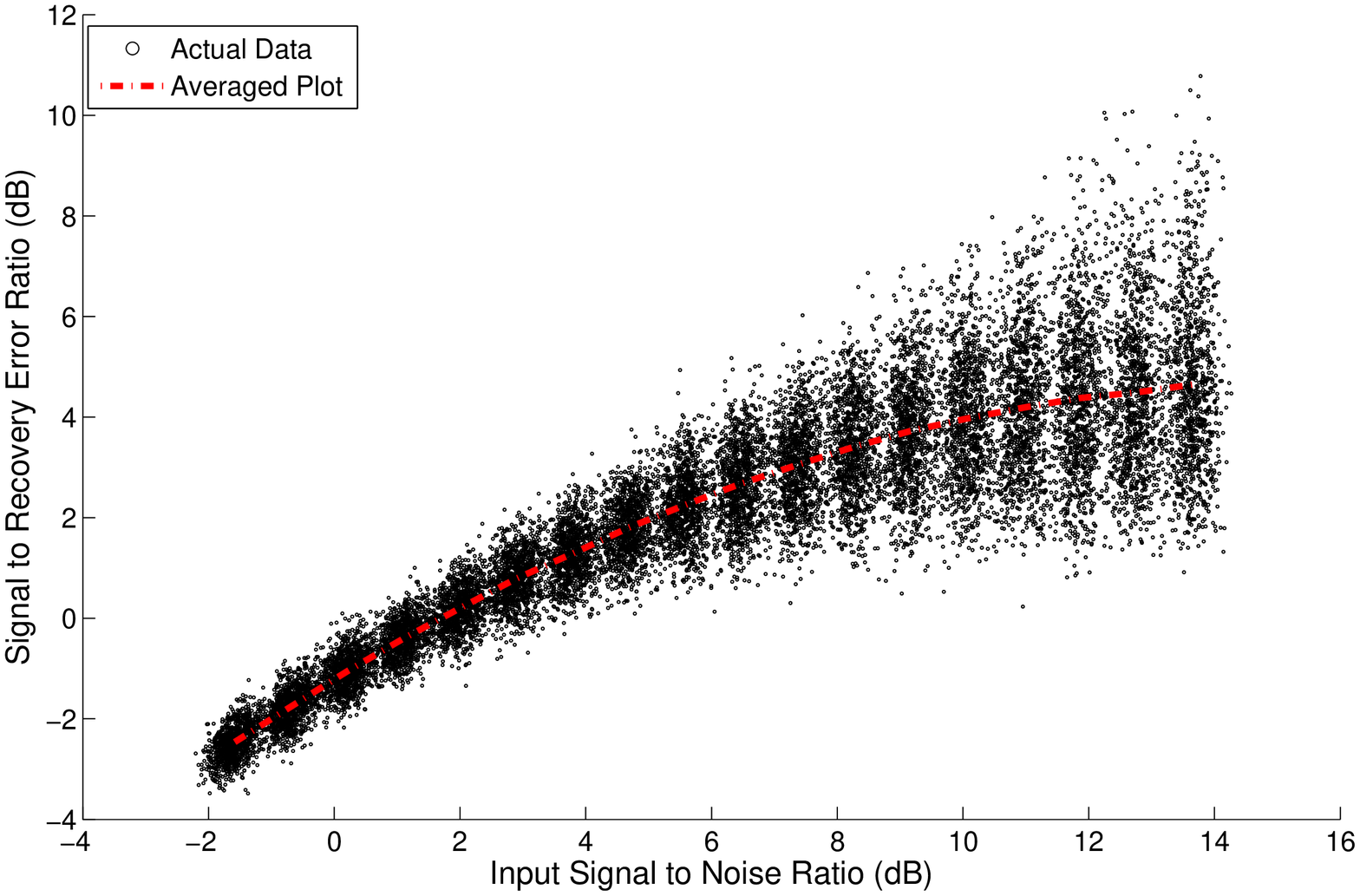}} \\
  \subfloat[$\omega=5$.]{\label{fig:avplotW=5}\includegraphics[width=0.4\textwidth]{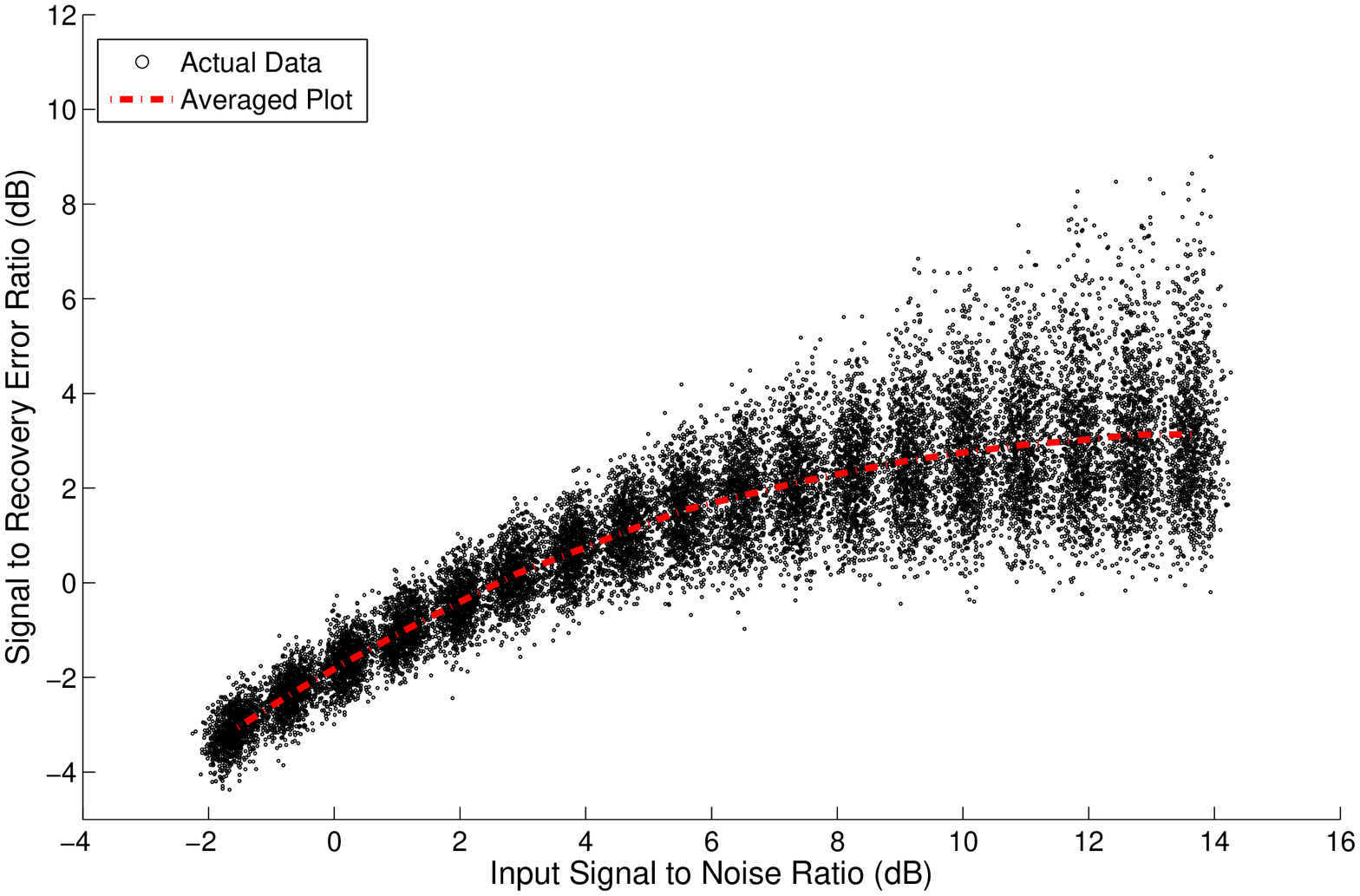}}
    \caption{ {\scriptsize Signal to recovery error ratio for weighted $\ell_1$ minimization with weight $\omega$  vs. input SNR for nonuniform sparse signals with $\gamma_1=\gamma_2=0.5$, $p_1=0.4,p_2=0.05$ superimposed with Gaussian noise.}}
  \label{fig:avplots}
\end{figure}

\begin{figure}[t]
\centering
  \includegraphics[width=0.6\textwidth]{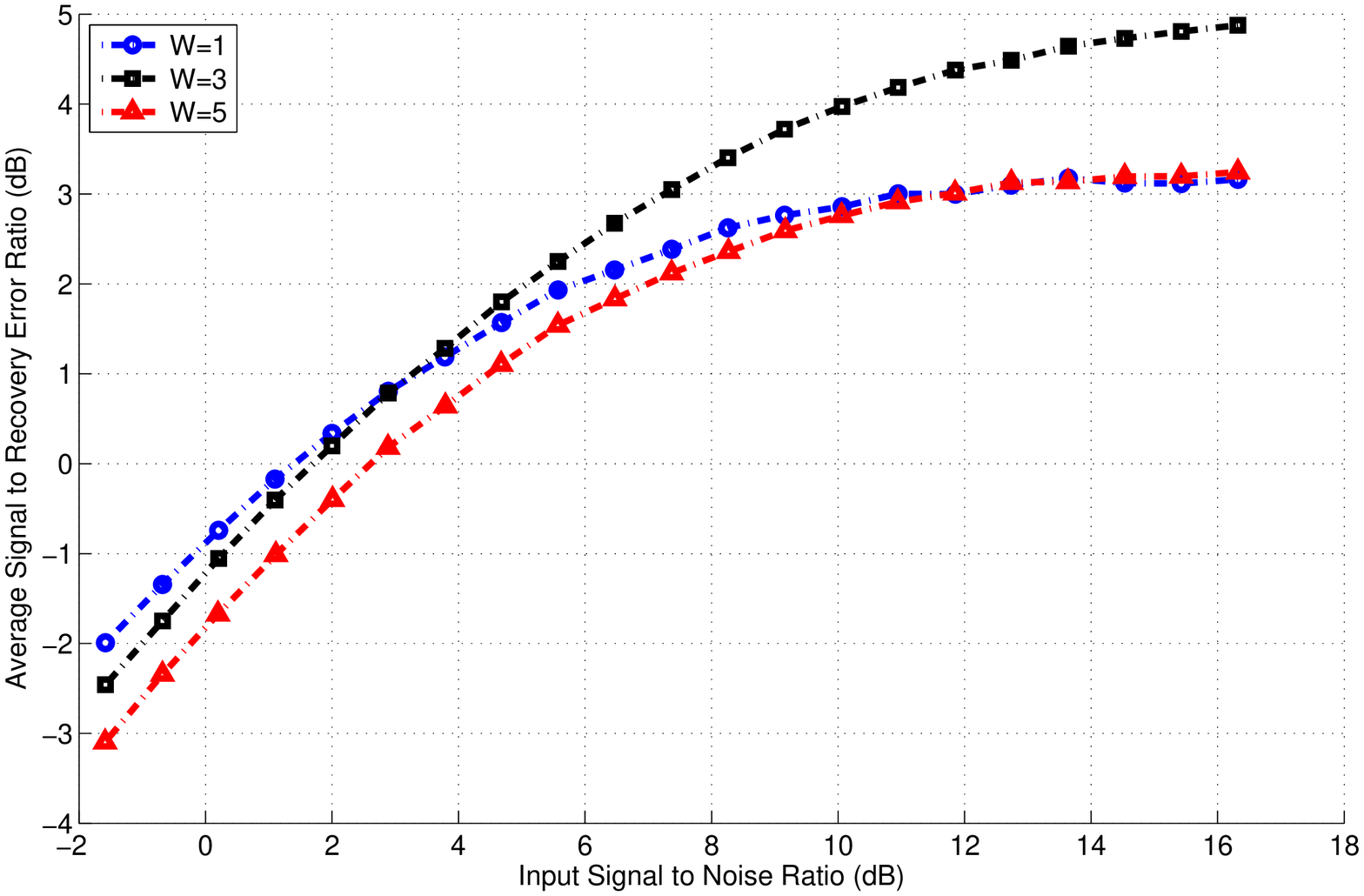}\\
\caption{{\scriptsize Average signal to recovery error ratio for weighted $\ell_1$ minimization with weight $\omega$  vs. input SNR for nonuniform sparse signals with $\gamma_1=\gamma_2=0.5$, $p_1=0.4,p_2=0.05$ superimposed with Gaussian noise. }} \label{fig:allavplots}
\end{figure}

\begin{figure}[t]
\centering
  \includegraphics[width= 0.35\textwidth]{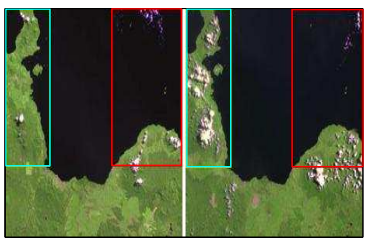}
  \caption{ \scriptsize{Satellite images taken from the New Britain rainforest in Papua Guina at 1989 (left) and 2000 (right). Image originally belongs to Royal Society for the Protection of Birds and was taken from the Guardian archive, an article on deforestation..}}
  \label{fig:sat image}
\end{figure}

\begin{figure}[t]
\centering
  \subfloat[]{\label{fig:sat1}\includegraphics[width= 0.4\textwidth]{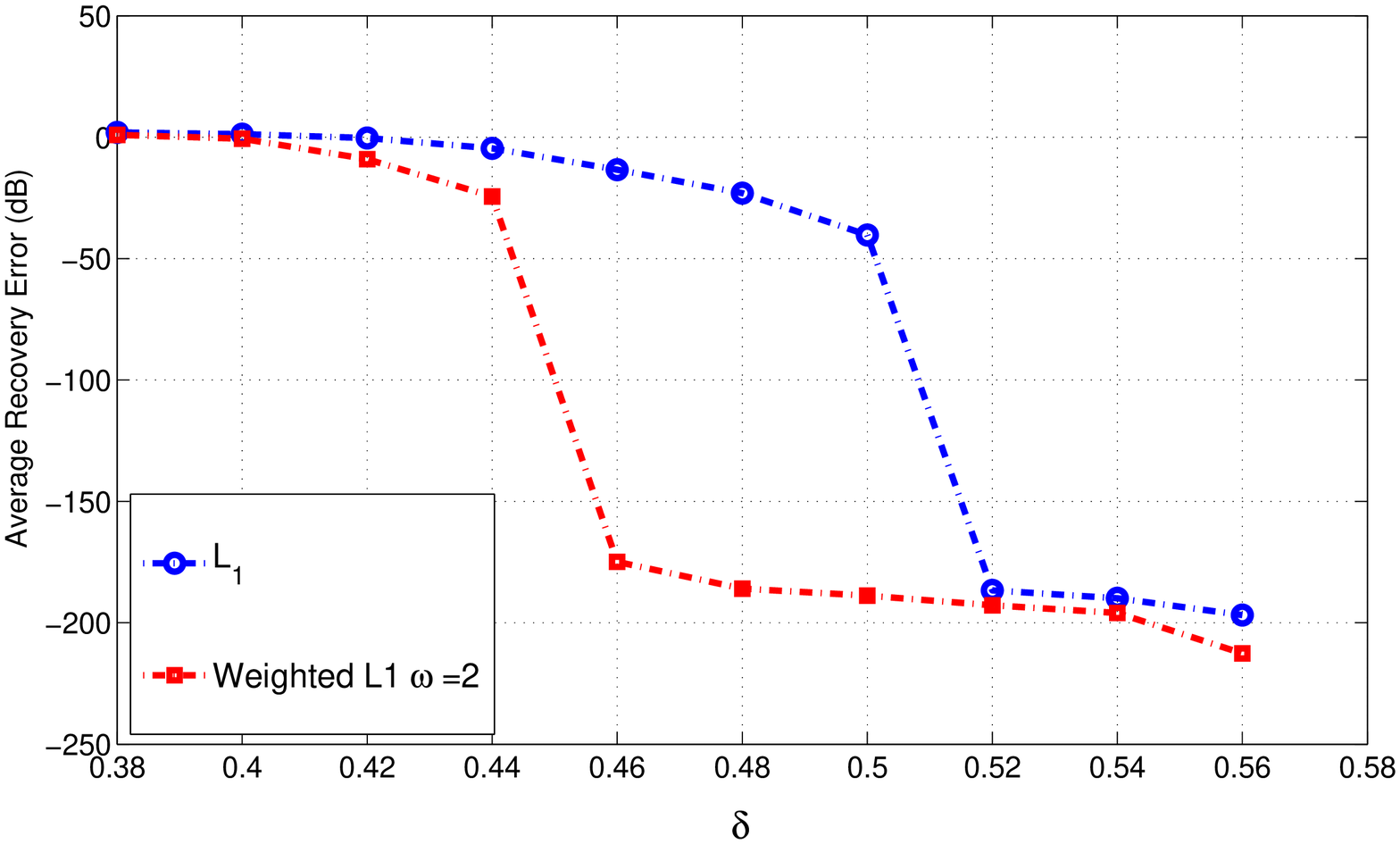}}
  \subfloat[]{\label{fig:sat2}\includegraphics[width= 0.4\textwidth]{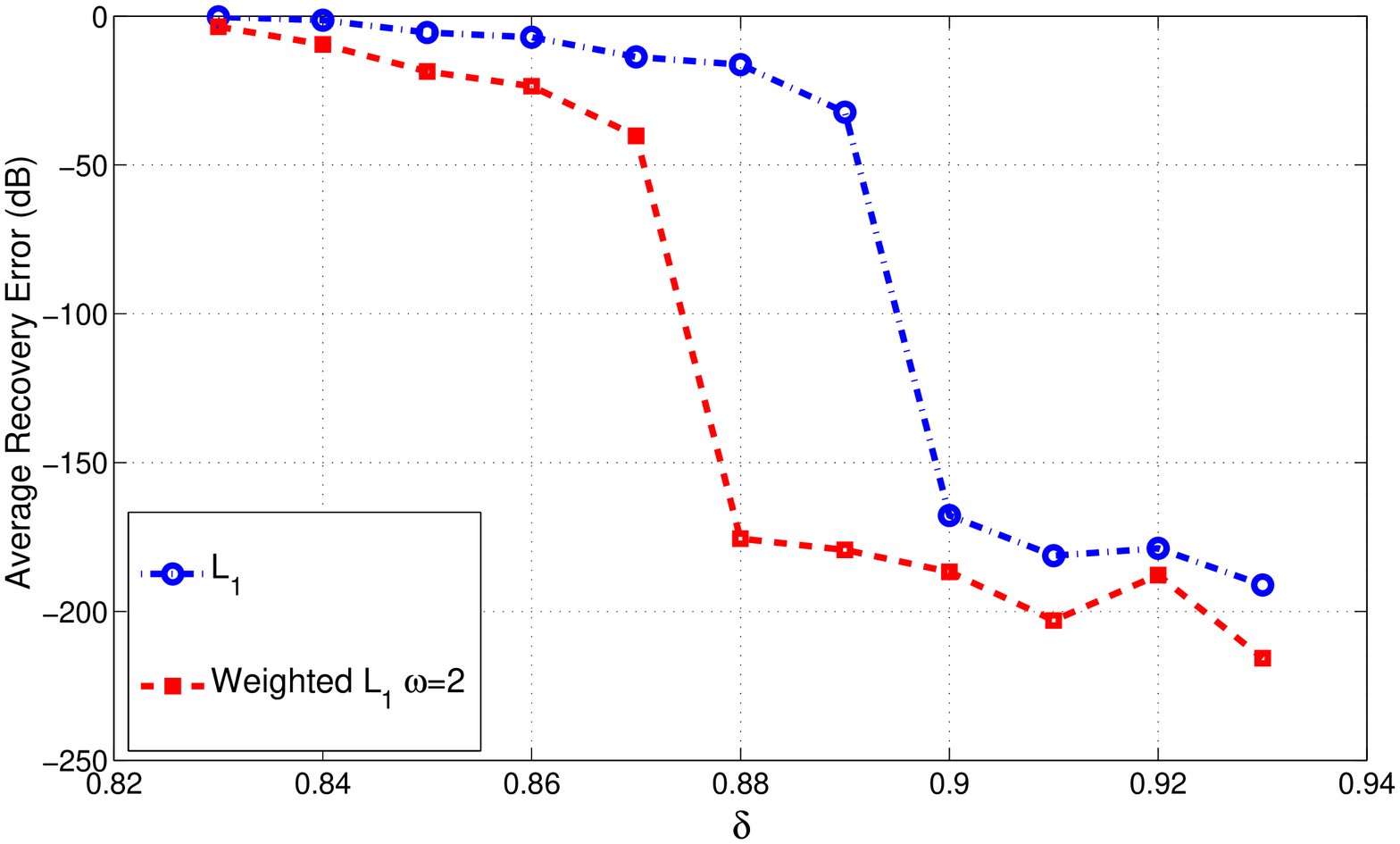}}
  \caption{ \scriptsize{Average Recovery Error for $\ell_1$ and weighted $\ell_1$ minimization recovery of the difference between the subframes of a pair of satellite images shown in Figure \ref{fig:sat image}. (a) corresponds to the right(red) subframe in Figure \ref{fig:sat image}, and (b) corresponds to the left(green) frame. Data is averaged over 50 realizations of measurement matrices for each $\delta$. }}
  \label{fig:simultions satellite}
\end{figure}

\begin{figure}[t]
\centering
  \includegraphics[width= 0.7\textwidth]{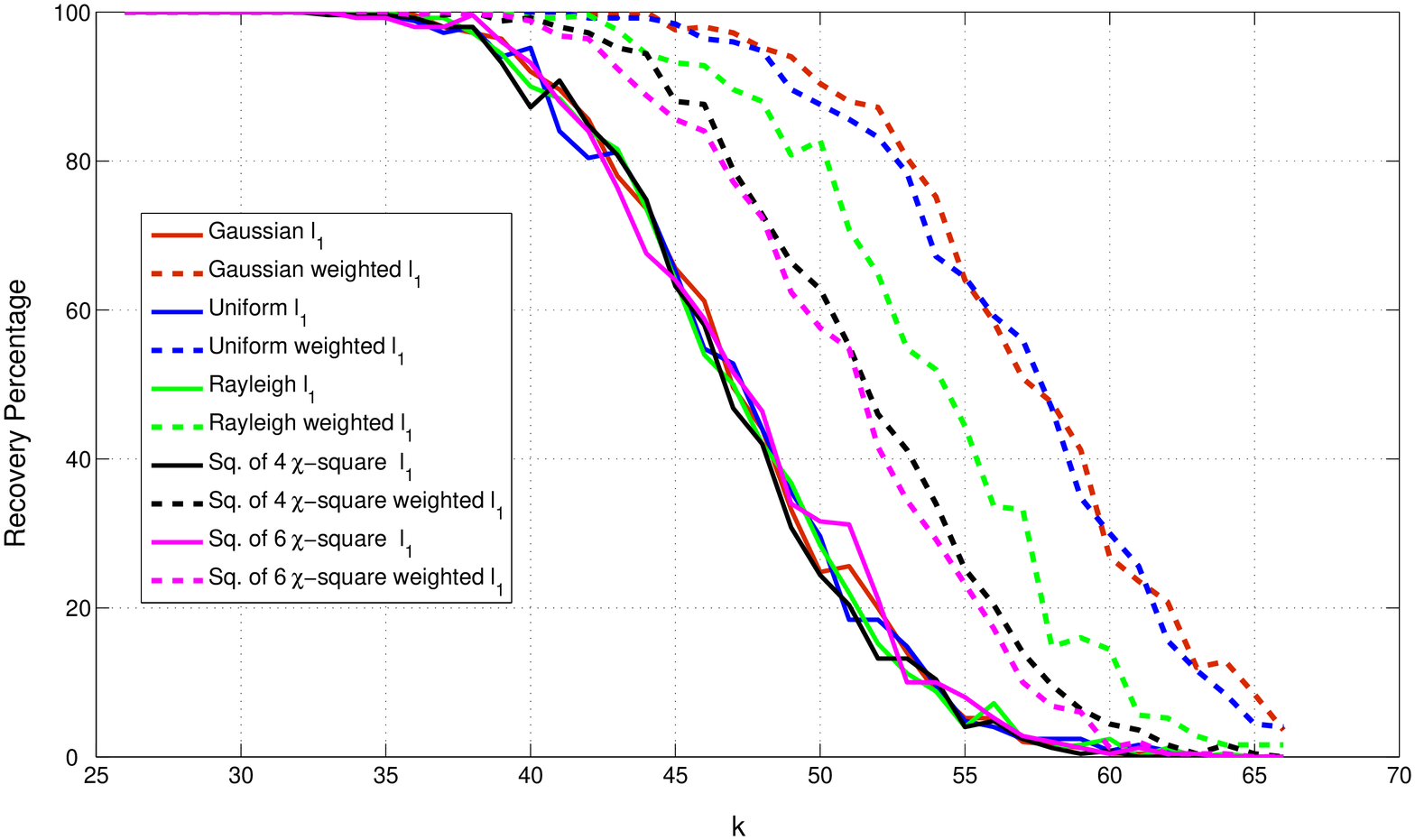}
  \caption{ \scriptsize{Empirical Recovery Percentage for Algorithm \ref{alg:modmain} with $n=200$ and $\delta = 0.5555$.}}
  \label{fig:simultions reweighted}
\end{figure}
We have done some experiments with regular $\ell_1$ and weighted $\ell_1$ minimization  recovery on some real world data. We have chosen a pair of satellite images (Figure \ref{fig:sat image})  taken at two different years, 1989 (left) and 2000 (right), from the New Britain rainforest in Papua Guina. These images are generally recorded to evaluate environmental effects such as \emph{deforestation}.  The difference of images taken at different times is generally not very significant, and thus can be thought of as compressible. In addition, the difference is usually more substantial over certain areas, e.g. forests. Therefore, it can be cast in a nonuniform sparse model. We have applied $\ell_1$ minimization  to recover the difference image over two subframes, identified by green and red rectangles in Figure \ref{fig:sat image}. In addition, a weighted $\ell_1$ minimization is also applied where the frame pixels are divided into two classes of equal sizes, where the concentration of the forestal area is larger over one of the classes, and hence the difference image is less sparse. For the right frame (red), the two classes are bottom half and top half of the frame, and for the left frame (green), they are left half and right half. We casually assign the weight value $\omega = 2$ for the sparser region for weighted $\ell_1$ recovery, and unitary weight to the denser region. The recovery errors for the two methods are displayed in Figure \ref{fig:simultions satellite}. The error is averaged over $50$ realizations of i.i.d. Gaussian measurement matrix for each $\delta$. As can be seen, even with this value of weight chosen intuitively, the recovery improvement is significant in the weighted $\ell_1$ minimization.

In figure \ref{fig:simultions reweighted}, we have compared the recovery performance for the regular $\ell_1$ minimization and the reweighted $\ell_1$ minimization of Algorithm \ref{alg:modmain}, for different sparsity levels and different distributions for the nonzero entries. Here the signal
dimension is $n=200$, and the number of measurements is $m=112$, which
corresponds to a value of $\delta = 0.5555$. We generated random
sparse signals with iid entries coming from certain
distributions; Gaussian, uniform,  Rayleigh , square root of
$\chi$-square with 4 degrees of freedom and, square root of
$\chi$-square with 6 degrees of freedom. Solid lines represent the simulation
results for ordinary $\ell_1$ minimization, and different colors
indicate different distributions. Dashed lines are used to show the
results for Algorithm
\ref{alg:modmain}. The reason why these distributions are selected and compared is elaborated in \cite{ISIT10 Reweighted}, as they demonstrate various levels of improvement. Note that for Gaussian and uniform distributions that are flat
and nonzero at the origin, the reweighted algorithm shows an impressive more than 20\%
improvement in the weak threshold (from 45 to 55).

\section{Conclusion and Future Work}
\label{sec:Conclusion} We analyzed the performance of the weighted
$\ell_1$ minimization for nonuniform sparse models. We computed
explicitly the phase transition curves for the weighted $\ell_1$
minimization, and showed that with proper weighting, the
recovery threshold for weighted $\ell_1$ minimization can be higher
than that of regular $\ell_1$ minimization. We provided simulation results to verify this both in the noiseless and noisy situation. Some of our simulations were performed on real world data of satellite images, where the nonuniform sparse model is a valid assumption. A further interesting question to be addressed in future work would be to characterize the gain in recovery percentage as a function of the number of
distinguishable classes $u$ in the nonuniform model. In
addition, we have used the results of this paper to build iterative
reweighted $\ell_1$ minimization algorithms that are provably
strictly better than $\ell_1$ minimization, when the nonzero entries
of the sparse signals are known to come from certain distributions
(in particular Gaussian distributions) \cite{ISIT10 Reweighted,KHXUAVHA Allerton}. The
basic idea there is that a simple post processing procedure on the
output of $\ell_1$ minimization results, with high probability, in a
hypothetical nonuniform sparsity model for the unknown signal, which
can be exploited for improved recovery.

\bibliographystyle{IEEEbib}

\appendix{}
\noindent {\Large \textbf{Appendix.  Proof of Important Lemmas }}
\section{Proof of Lemma \ref{lemma:baselemma}}
\label{App: proof of lemma baselemma} First, let us assume that
$\sum_{i\in K}w_i|z_i|\leq \sum_{i\in\overline{K}}w_i|z_i|, \forall
\z \in \mathcal{N}(A)$. Note that by assumption $w_i$s are all
nonnegative. Using the triangular inequality for the weighted
$\ell_1$ norm (or for each absolute value term on the LHS) we obtain
\begin{eqnarray*}
\sum_{i\in K} w_i|x_i+z_i|+ \sum_{i\in\overline{K}}w_i|z_i| &\geq& \sum_{i\in K} w_i|x_i|-\sum_{i\in K}w_i|z_i|+ \sum_{i\in\overline{K}}w_i|z_i| \nonumber \\
&\geq& \sum_{i\in K} w_i|x_i| . \nonumber
\end{eqnarray*}
thus proving the forward part of this lemma. Now let us assume
instead that $\exists \z\in \mathcal{N}(A)$, such that $\sum_{i\in
K}w_i|z_i| > \sum_{i\in\overline{K}}w_i|z_i|$. Then we can construct
a vector $\x$ supported on the set $K$ (or a subset of $K$), with
$\x_{K}=-\z_{K}$ (i.e. $x_i = -z_i~ \forall i\in K$). Then we have
\begin{eqnarray*}
\sum_{i\in K} w_i|x_i+z_i|+ \sum_{i\in\overline{K}}w_i|z_i| = 0 +
\sum_{i\in\overline{K}}w_i|z_i| < \sum_{i\in K} w_i|x_i|
\end{eqnarray*}
proving the reverse part of this lemma.

\section{Proof of Lemma \ref{lemma:Con_F,G}}
\label{App:proof of lemma Con_F,G} Without loss of generality,
assume that $\F$ has the following $k$ vertices:
$\{\frac{\e_r}{w_r}, 1\leq r \leq k\}$, where $\e_r$ is the
$n$-dimensional standard unit vector with the $r$-th element equal
to 1. Also assume that the $(l-1)$-dimensional face $G$ is the
convex hull of the following $l$ vertices: $\frac{\e_r}{w_r}, 1\leq
r \leq l$. Then the cone $\mathcal{C}_{\F,\G}$ formed by observing
the $(l-1)$-dimensional face $\G$ of the weighted $\ell_1$-ball
$\mathcal{P}_{\w}$ from an interior point $x^{\F}$ of the face $\F$
is the positive cone of the vectors:
\begin{equation}
\frac{\e_j}{w_j}-\frac{\e_i}{w_i}, ~~\text{for all}~~ j \in {J
\backslash K},~i \in {K}, \label{eq:stset}
\end{equation}
and also the vectors
\begin{equation}
\frac{\e_{i_1}}{w_{i_1}}-\frac{\e_{i_2}}{w_{i_2}}, ~~\text{for
all}~~{i_1} \in {K},~{i_2} \in {K}, \label{eq:ndset}
\end{equation}
where $L=\{1, 2,..., l\}$ is the support set for the face $\G$. So
the cone $\mathcal{C}_{\F,\G}$ is the direct sum of the linear hull
$\mathcal{L}_{\F}=\text{lin}\{\F-x^{\F}\}$ formed by the vectors in
(\ref{eq:ndset}) and the cone
$\mathcal{C}_{\F^{\perp},G}=\mathcal{C}_{F,G} \bigcap
\mathcal{\mathcal{L}}_{\F}^{\perp}$, where $\mathcal{L}_{F}^{\perp}$
is the orthogonal complement to the linear subspace
$\mathcal{L}_{\F}$. Then $\mathcal{C}_{\F^{\perp},\G}$ has the same
(relative) spherical volume as $\mathcal{C}_{\F,\G}$, and by
definition the internal angle $\beta(\F,\G)$ is the relative
spherical volume of the cone $\mathcal{C}_{\F,\G}$. Now let us
analyze the structure of $\mathcal{C}_{\F^{\perp},\G}$. We notice
that the vector $\e_{0}=\sum_{r=1}^{k}w_r\e_r$ is in the linear
space $\mathcal{L}_{\F}^{\perp}$ and is also the only such a vector
(up to linear scaling) supported on $K$. Thus a vector $\x$ in the
positive cone $\mathcal{C}_{\F^{\perp},\G}$ must take the form
\begin{equation}
-\sum_{i=1}^{k}{b_i   e_i}+\sum_{i=k+1}^{l}{b_i   e_i},
\label{eq:vform}
\end{equation}
where $b_i, 1 \leq i \leq l$ are nonnegative real numbers  and
\begin{eqnarray*}
&& \sum_{i=1}^{k}{w_i b_i}=\sum_{i=k+1}^{l}{w_i
b_{i}}~~~\frac{b_1}{w_{K_1}}=\frac{b_2}{w_{K_2}}=\cdots=\frac{b_k}{w_k}.
\end{eqnarray*}

Now that we have identified $\mathcal{C}_{\F^{\perp},\G}$ we try to
calculate its relative spherical volume with respect to the sphere
surface $S^{l-k-1}$ to derive $\beta(\F,\G)$.  First, we notice that
$\mathcal{C}_{\F^{\perp},\G}$ is a $(l-k)$-dimensional cone. Also,
all the vectors $(x_1, \cdots, x_{n})$ in the cone
$\mathcal{C}_{\F^{\perp},\G}$ take the form in (\ref{eq:vform}).
From \cite{Hadwiger},
\begin{eqnarray}
\nonumber\int_{\mathcal{C}_{\F^{\perp},\G}}{e^{-\|\x\|^2}}\,d\x=\beta(\F,\G)
V_{l-k-1}(S^{l-k-1})\int_{0}^{\infty}{e^{-r^2}} r^{l-k-1}\,dx
=\beta(\F,\G) \cdot \pi^{(l-k)/2}, \label{eq:inaxchdirect1}
\end{eqnarray}
where $V_{l-k-1}(S^{l-k-1})$ is the spherical volume of the
$(l-k-1)$-dimensional sphere $S^{l-k-1}$ and is given by the
well-known formula
\begin{equation*}
V_{i-1}(S^{i-1})=\frac{i\pi^{\frac{i}{2}}}{\Gamma(\frac{i}{2}+1)},
\end{equation*}
where $\Gamma(\cdot)$ is the usual Gamma function. This completes
the proof.

\section{Proof of Lemma \ref{lemma:extasy}}
\label{App: proof of lemma ext. exp.} Let $G$ denote the cumulative
distribution function of a half-normal $HN(0, 1/2)$ random variable,
i.e. a random variable $X = |Z|$ where $Z \sim N(0, 1/2)$, and $G(x)
= \Prob \{X \leq x\}$. Since $X$ has density function $g(x) =
\frac{2}{\sqrt{\pi}}\exp(-x^2)$, we know that
\begin{equation}
G(x) = \frac{2}{\sqrt{\pi}} \int_{0}^{x}e^{-y^2}\,dy; \label{eq:erf}
\end{equation}
 and so $G$ is just the classical error function \emph{erf($\cdot$)}. We now
justify the external angle exponent computations in Theorem
{\ref{thm:main delta bound}} and Lemma \ref{lemma:extasy} using
Laplace methods \cite{D}.  Using the same set of notations as in
Theorem $\ref{thm:main delta bound}$, let $t_1=\tau_1n$,
$t_2=\tau_2n$. Also define $c=(\tau_1+\gamma_1p_1)+\omega
^2(\tau_2+\gamma_2p_2)$, $\alpha_1=\gamma_1(1-p_1)-\tau_1$ and
$\alpha_2=\gamma_2(1-p_2)-\tau_2$. Let $x_0$ be the unique solution
to $x$ of the following:
\begin{equation}
2c-\frac{g(x)\alpha_1}{xG(x)}-\frac{\omega g(\omega
x)\alpha_2}{xG(\omega x)}=0 \label{eq:externalweiyu}
\end{equation}
Since $xG(x)$ is a smooth strictly increasing function ( $\sim 0$ as
$x \rightarrow 0$ and $\sim x$ as $x\rightarrow \infty$), and $g(x)$
is strictly decreasing, the function
$\frac{g(x)\alpha_1}{xG(x)}+\frac{\omega g(\omega
x)\alpha_2}{xG(\omega x)}$ is one-one on the positive axis, and
$x_0$ is a well-defined function of $\tau_1$ and $\tau_2$. Hence, we
denote it as $x_0(\tau_1,\tau_2)$. Then
\begin{equation}
\psi_{ext}(\tau_1,\tau_2) =
cx_0^2-\alpha_1\log{G(x_0)}-\alpha_2\log{G(\omega x_0)}.
\end{equation}
\noindent To prove Lemma \ref{lemma:extasy}, we start from the
explicit integral formula
\begin{equation}
\zeta(d_1,d_2)=\pi^{-\frac{n-l+1}{2}}2^{n-l}\int_{0}^{\infty}e^{-x^2}
\left (\int_{0}^{ \frac{w_{K_1}x}{ \xi(d_1,d_2)} } e^{-y^2}\,dy
\right)^{r_1} \left (\int_{0}^{ \frac{w_{K_2} x}{\xi(d_1,d_2)} }
e^{-y^2}\,dy \right)^{r_2}\,dx,
\end{equation}

\noindent After a changing of integral variables (Noticing that
$w_{K_1}=1$, $w_{K_2}=\omega $, $\frac{n_{1}-d_{1}}{n}=\alpha_1$,
and $\frac{n_{2}-d_{2}}{n}=\alpha_2$ ), we have
\begin{eqnarray}
\zeta(t_1+k_1,t_2+k_2)=\sqrt{cn/\pi}
\int_{0}^{\infty}e^{-n(cx^2-\alpha_{1}\log(G(x))-\alpha_{2}\log(G(\omega
x))} \,dx.
\end{eqnarray}
 This suggests that we should use Laplace's method; we define
\begin{equation}
f_{\tau_1,\tau_2,n}=e^{-n \psi_{t_{1}',t_{2}'}(y)} \cdot
\sqrt{cn/\pi} \label{eq:fnnu}
\end{equation}
with
\begin{equation*}
\psi_{t_{1}',t_{2}'}(y)=cy^2-c_{1}\log G(y)-\alpha_{2}\log G(\omega
y)
\end{equation*}

\noindent We note that the function $\psi_{t_{1}',t_{2}'}$ is smooth
and convex. Applying Laplace's method to $\psi_{t_{1}',t_{2}'}$, but
taking care about regularity conditions and remainders as in
\cite{D}, gives a result with the uniformity in $(t_{1}',t_{2}')$.

\begin{lem}
For $t_{1}', t_{2}'$, let $x_{0}(\tau_1,\tau_2)$ denote the
minimizer of $\psi_{t_{1}',t_{2}'}$. Then
\begin{equation*}
\int_{0}^{\infty} f_{t_{1}',t_{2}',n}(x) \,dx \leq
e^{-n\psi_{t_{1}',t_{2}'}(x_{0}(t_{1}',t_{2}'))}
(1+R_n(t_{1}',t_{2}')),
\end{equation*}
where for any $\delta, \eta>0$,
\begin{equation*}
\sup_{0 \leq t_1^{'} \leq \gamma_{1}-\rho_{1}, \\ 0\leq \tau_2 \leq
(\gamma_{2}-\rho_{2}), \\
\delta-\rho_1-\rho_{2}\leq \tau_1+\tau_2 \leq
(1-\rho_1-\rho_{2}-\eta) } R_n(t_{1}',t_{2}')=o(1)~~as~~n\rightarrow
\infty.
\end{equation*}
\noindent where $\rho_{1}=k_{1}/n$, $\rho_{2}=k_{2}/n$,
$n_{1}/n=\gamma_{1}$ and $n_{2}/n=\gamma_{2}$. \label{lemma:extlap}
\end{lem}
\noindent In fact, in this lemma, the minimizer
$x_{0}(t_{1}',t_{2}')$ is exactly the same $x_{0}(t_{1}',t_{2}')$
defined earlier in (\ref{eq:externalweiyu}) and the corresponding
minimum value is the same as the defined exponent $\psi_{ext}$:
\begin{equation}
\psi_{ext}(t_{1}',
t_{2}')=\psi_{t_{1}',t_{2}'}(x_{t_{1}'},x_{t_{2}'}).
\label{eq:bigsmall}
\end{equation}
\noindent We can derive Lemma \ref{lemma:extasy} from Lemma
\ref{lemma:extlap}. We note that as
$t_{1}'+t_{2}'+\gamma_{1}+\gamma_{2} \rightarrow 1$,
$x_{0}({t_{1}',t_{2}'}) \rightarrow 0$ and
$\psi_{ext}(t_{1}',t_{2}') \rightarrow 0$. For given $\epsilon
>0$ in the statement of Lemma \ref{lemma:extasy}, there is a largest
$\eta_{\epsilon} <1$ such that as long as
$\tau_1+\tau_2+\rho_{1}+\rho_{2}
>\eta_{\epsilon}$, $\psi_{ext}(t_{1}',t_{2}')<\epsilon$. Note that
$\zeta(\G,\mathcal{P}_{\w}) \leq 1$, so that for
$\tau_1+\tau_2+\rho_{1}+\rho_{2}
>\eta_{\epsilon}$,
\begin{equation*}
n^{-1}\log(\zeta(t_1+k_1,t_2+k_2))\leq 0<
-\psi_{ext}(t_{1}',t_{2}')+\epsilon,
\end{equation*}
for $n \geq 1$.  Applying the uniformity in $t_{1}',t_{2}'$ given in
Lemma \ref{lemma:extlap}, we have as $n \rightarrow \infty$,
uniformly over the feasible region for $t_{1}',t_{2}'$,
\begin{equation}
n^{-1}\log(\zeta(t_1+k_1,t_2+k_2)) \leq
-\psi_{ext}(t_{1}',t_{2}')+o(1).
\end{equation}
Then Lemma \ref{lemma:extasy} follows.

\section{Proof of Lemma \ref{lemma:intasy}}
\label{App: proof of lemma ext. int.}

Recall Theorem \ref{thm:main internal angle thm}. By applying the
large deviation techniques as in \cite{D}, we have
\begin{equation}
p_{Z}(0) \leq \frac{2}{\sqrt{\pi\Omega}}
 \left( \int_{0}^{\mu_{m'}} {v
e^{-v^2-m'\Lambda^*(\frac{\sqrt{2\Omega}}{m'}v)}}\,dv+e^{-\mu_{m'}^2}
\right), \label{eq:pz}
\end{equation}
where $\Omega$ is the same as defined in Section \ref{sec:
Derivation of Inter.}, $w_{K_1}=1$, $w_{K_2}=\omega$, $m'=t_1+t_2$,
$\mu_{m'}=(t_1+t_2\omega)\sqrt{\frac{1}{\pi\Omega}}$ is the
expectation of
$\frac{1}{\sqrt{\Omega'}}(w_{K_1}^2\sum_{i=1}^{t_1}X_i'-w_{K_2}^2\sum_{i=1}^{t_2}X_i'')$,
($X_i'$ and  $X_i''$ are defined as in Theorem \ref{thm:main
internal angle thm}), and
\begin{equation*}
\Lambda^*(y)={\max_{s}}~sy-\frac{t_1}{t_1+t_2}\Lambda_{1}(s)-\frac{t_2}{t_1+t_2}\Lambda_{2}(s),
\end{equation*}
with
\begin{equation*}
\Lambda_{1}(s)=\frac{s^2}{2}+\log(2\Phi(s)),~~\Lambda_{2}(s)=\Lambda_{1}(\omega
s).
\end{equation*}
In fact, the second term in the sum can be argued to be negligible
\cite{D}. After a changing of variables
$y=\frac{\sqrt{2\Omega}}{m'}v$, we know that the first term of
(\ref{eq:pz}) is upper-bounded by
\begin{equation}
\frac{2}{\sqrt{\pi}}\cdot\frac{1}{\sqrt{\Omega}}\cdot\frac{m'^2}{2\Omega}
\cdot \int_{0}^{\frac{t_1+t_2\omega}{t_1+t_2}\sqrt{2/\pi}} {y e^{-m'
(\frac{m'}{2\Omega})y^2-m'\Lambda^*(y)}}\,dy. \label{eq:intld}
\end{equation}

As we know, $m'$ in the exponent of (\ref{eq:intld}) is $t_1+t_2$.
Similar to evaluating the external angle decay exponent, we will
resort to the Laplace's method in evaluating the internal angle
decay exponent.

\noindent Define the function
\begin{equation*}
 f_{t_1,t_2}(y)=y e^{-m'
(\frac{m'}{2\Omega})y^2-m'\Lambda^*(y)}.
\end{equation*}
If we apply similar arguments as in proving Lemma \ref{lemma:extlap}
and take care of the uniformity, we have the following lemma:
\begin{lem}
Let $y_{t_1,t_2}*$ denotes the minimizer of
$(\frac{m'}{2\Omega})y^2+\Lambda^*(y)$. Then
\begin{equation*}
\int_{0}^{\infty} f_{t_1,t_2}(x) \,dx \leq
e^{-m'\left((\frac{m'}{2\Omega})y_{t_1,t_2}*^2 +\Lambda^*(
y_{t_1,t_2}* )\right)} \cdot R_{m'}(t_1,t_2)
\end{equation*}
where for $\eta>0$
\begin{equation*}
m'^{-1} \sup_{t_1,t_2} \log(R_{m'}(t_1,t_2))=o(1)~~as~~m'\rightarrow
\infty.
\end{equation*}
\end{lem}

\noindent This means that
\begin{equation*}
p_{Z}(0) \leq e^{-m'\left((\frac{m'}{2\Omega})y_{t_1,t_2}*^2
+\Lambda^*( y_{t_1,t_2}* )\right)} \cdot R_{m'}(t_1,t_2),
\end{equation*}
where $m'^{-1} \sup_{\frac{(t1+t_2)}{n} \in [\delta-\rho_1-\rho_2,
1]} \log(R_{m'}(t_1,t_2))=o(1)~~as~~m'\rightarrow \infty$.

Now in order to find a lower bound on the decay exponent for
$p_Z(0)$,(ultimately the decay exponent
$\psi_{int}(\tau_1,\tau_2)$), we need to focus on finding the
minimizer $y_{t_1,t_2}*$ for $(\frac{m'}{2\Omega})y^2+\Lambda^*(y)$.
On this way, by setting the derivative of
$(\frac{m'}{2\Omega})y^2+\Lambda^*(y)$ with respect to $y$ to 0, and
also noting the derivative ${\Lambda^{*}}'(y)=s$, we have
\begin{equation}
s=-\frac{m'}{\Omega}y. \label{eq:sy2}
\end{equation}

\noindent At the same time, the $s$ maximizing $\Lambda^{*}(y)$ must
satisfy
\begin{equation}
y=\frac{t_1}{t_1+t_2}
\Lambda_{1}'(s)+\frac{t_2}{t_1+t_2}\Lambda_{2}'(s), \label{eq:ys2}
\end{equation}

\noindent namely, (by writing out (\ref{eq:ys2})),
\begin{equation}
y=\frac{t_1+\omega^2t_2}{t_1+t_2}s+Q(s),
\end{equation}
where $Q(s)$ is defined as in Theorem \ref{thm:main internal angle
thm}. By combining (\ref{eq:sy2}) and (\ref{eq:ys2}), we can solve
for the $s$ and $y$, thus resulting in the decay exponent for
$\psi_{int}(\tau_1,\tau_2)$ as calculated in Theorem \ref{thm:main
internal angle thm}
\section{Proof of Theorem \ref{thm:p1=1 theorem}}
\label{App: proof of theorem p1=1}
Let $\delta' = \delta_c(\gamma_1,\gamma_2,1,p_2,\omega)$ and $\delta'' = \delta_c(0,1,0,p_2,1)$. From Theorem \ref{thm:main delta bound} we know that:
\bea \delta' &=& \min\{\delta~|~\psi'_{com}(0,\tau_2)-\psi'_{int}(0,\tau_2)-\psi'_{ext}(0,\tau_2)<0~ \forall~
0\leq \tau_2\leq \gamma_2(1-p_2)\nonumber\\
&&~~~~~~~~, \tau_2 > \delta-\gamma_1-\gamma_2p_2 \},\nonumber \\
   &=& \gamma_2\min\{\delta~|~\psi'_{com}(0,\gamma_2\tau_2)-\psi'_{int}(0,\gamma_2\tau_2)-\psi'_{ext}(0,\gamma_2\tau_2)<0~ \forall~
0\leq \tau_2\leq 1-p_2\nonumber\\
&&~~~~~~~~, \tau_2 > \delta-p_2 \}+\gamma_1,\label{eq:aux1}
\eea
\noindent and
\bea \delta'' &=& \min\{\delta~|~\psi''_{com}(0,\tau_2)-\psi''_{int}(0,\tau_2)-\psi''_{ext}(0,\tau_2)<0~ \forall~
0\leq \tau_2\leq 1-p_2\nonumber\\
&&~~~~~~~~, \tau_2 > \delta-p_2 \},\label{eq:aux2}\eea
\noindent where the exponents $\psi'_{com}$,$\psi'_{int}$,$\psi'_{ext}$,$\psi''_{com}$,$\psi''_{int}$ and $\psi''_{ext}$ can be found using Theorem \ref{thm:main delta bound}. Here, we basically show that when $\omega\rightarrow \infty$:
\bea \psi'_{com}(0,\gamma_2\tau_2)&=&\gamma_2\psi''_{com}(0,\tau_2),\label{eq:1}\\
\psi'_{int}(0,\gamma_2\tau_2)&=&\gamma_2\psi''_{int}(0,\tau_2),\label{eq:2}\\ \psi'_{ext}(0,\gamma_2\tau_2)&=&\gamma_2\psi''_{ext}(0,\tau_2).\label{eq:3}\eea

\noindent (\ref{eq:1}) follows immediately from the definition of $\psi_{com}$ in (\ref{eq:comb angle}). On the other hand, from (\ref{eq:ext angle}), for $\omega\rightarrow$ we know that
\bea \psi'_{ext}(0,\gamma_2\tau_2) =
c'{x'_0}^2-\alpha'_2\log{G(\omega x'_0)}, \nonumber\\
\psi''_{ext}(0,\tau_2) =
c''{x''_0}^2-\alpha''_2\log{G(x''_0)}.\nonumber
\eea
\noindent  Following the details of derivations as in Theorem \ref{thm:main delta bound}, we realize that:
\beq c'=\gamma_2\omega^2c'',~\omega x'_0 = x''_0,~\alpha'_2 = \gamma_2\alpha''_2, \eeq
\noindent which implies that  $\psi'_{ext}(0,\gamma_2\tau_2)=\gamma_2\psi''_{ext}(0,\tau_2)$. Finally, from (\ref{eq:int angle}), we know that
\bea
\psi'_{int}(0,\gamma_2\tau_2) =
(\Lambda^*(y')+\frac{\gamma_2\tau_2}{2\Omega'}y'^2+\log2)\gamma_2\tau_2,\nonumber \\
\psi''_{int}(0,\tau_2) =
(\Lambda^*(y'')+\frac{\tau_2}{2\Omega''}y''^2+\log2)\tau_2.\nonumber
\eea
\noindent  Following the details of derivations as in Theorem \ref{thm:main delta bound}, we realize that for $\omega\rightarrow\infty$:
\beq y'=y'',~\Omega'=\gamma_2\Omega''.\eeq
\noindent which implies that  $\psi'_{int}(0,\gamma_2\tau_2)=\gamma_2\psi''_{int}(0,\tau_2)$. From (\ref{eq:aux1}), (\ref{eq:aux2}) and (\ref{eq:1})-(\ref{eq:3}) it follows that

\beq \delta' = \gamma_2\delta'' + \gamma_1.\eeq

\end{document}